\documentclass{amsart}
\usepackage{srcltx}
\usepackage[usenames]{color}
\usepackage{amssymb,amsmath}
\usepackage{graphicx}
\newtheorem{thm}{Theorem}[section]
\newtheorem{remm}{Remark}[section]
\newtheorem{lem}{Lemma}[section]

\newtheorem{remark}[remm]{Remark}

\numberwithin{equation}{section}
\newcommand{\ssy}{\scriptscriptstyle}
\newcommand{\ld}{\left<}
\newcommand{\rd}{\right>}
\begin{document}
\title[]{Conservation laws, exact travelling waves 
and modulation instability 
for an extended Nonlinear Schr\"{o}dinger equation}
{}
\author[]{V. Achilleos$^{1}$, S. Diamantidis$^{2}$, D. J. Frantzeskakis$^{1}$, N. I. Karachalios$^{2}$
and P. G. Kevrekidis$^{3}$}
%
\thanks
{$^{1}$ Department of Physics, University of Athens, Panepistimiopolis, Zografos, 
Athens 15784, Greece.
}
\thanks
{$^{2}$ Department of Mathematics, University of the Aegean, Karlovassi, 
83200, Samos, Greece.
}
\thanks
{$^{3}$ Department of Mathematics and Statistics, University of Massachusetts, 
Amherst MA 01003-4515, USA
}
%
%
%
%
\begin{abstract}
We study various properties of solutions of an extended nonlinear  Schr\"{o}dinger (ENLS) equation, which  
arises in the context of geometric evolution problems -- 
including vortex filament dynamics -- and governs  
propagation of short pulses in optical fibers and nonlinear metamaterials. 
For the periodic initial-boundary value problem, 
we derive conservation laws satisfied by local in time, weak $H^2$ (distributional) solutions, and 
establish global existence of such weak solutions. 
The derivation is obtained by a regularization scheme under a balance condition on the coefficients 
of the linear and nonlinear terms -- namely, the Hirota limit of the considered ENLS model. 
Next, we investigate conditions for the existence of traveling wave solutions, 
focusing on the case of bright and dark solitons. 
The balance condition on the coefficients is found to be 
essential for the existence of exact analytical 
soliton solutions; furthermore, we obtain conditions which define parameter regimes for the 
existence of traveling solitons for various 
linear dispersion strengths. Finally, we study 
the modulational instability of plane waves of the ENLS equation, and identify 
important differences between the ENLS case and the corresponding NLS counterpart. 
The analytical results are corroborated by numerical simulations, which reveal notable differences 
between the bright and the dark soliton 
propagation dynamics, and are in excellent agreement with the analytical predictions of the modulation instability analysis.
\end{abstract}
\maketitle
%
%
%
%
\section{Introduction}
\subsection{The model and its physical significance}
The present paper deals with an extended 
nonlinear Schr\"{o}dinger (ENLS) equation, expressed in the following dimensionless form:
\begin{equation}\label{introeq1}
\phi_t+3\alpha|\phi|^2\phi_x-{\rm
i}\,\rho\,\phi_{xx}+\sigma\,\phi_{xxx} -{\rm
i}\,\delta |\phi|^2\,\phi=0, 
\end{equation}
%
where $\phi(x,t)$ is the unknown complex field, subscripts denote partial derivatives, and 
$\alpha$, $\rho$, $\sigma$ and $\delta$ are real constants. 
Equation (\ref{introeq1}) has important applications in distinct physical and mathematical contexts. 
First, from the physical point of view, this model is one of the nonlinear dispersive equations 
which under the Hasimoto thransformation, $\phi(x,t)=\kappa(x,t)\,\exp\left({\rm i}\,\int_0^{x}
\tau(\sigma,t)\,d\sigma\right)$ can be used to determine 
the motion of vortex filaments 
\cite{Chorin,Chorin2,Nsol1,JFM1991,Hasi,Hopfin1,Kolmogorov,Hopfin2}: here, 
the functions $\kappa(x,t)$ and $\tau(x,t)$ denote the curvature and torsion of the evolving
filament, respectively, while $x$ represents length measured along the filament. When $\kappa(x,t)$ and 
$\tau(x,t)$ are specified, the shape of the filament as a curve is uniquely determined as it evolves in space \cite{Hasi}. 

It is crucial to remark that the derivation of (\ref{introeq1}) as an evolution equation for the Hasimoto transformation is valid only in the case where the coefficients $\alpha, \rho,\sigma,\delta$
satisfy the balance condition:
\begin{equation}\label{cruc2008}
\alpha\rho=\sigma\delta.
\end{equation}
Then, Eq.~\eqref{introeq1} is known as \emph{Hirota equation} 
\cite{Hirotar}, which is a combination of the nonlinear Schr{\"o}dinger (NLS) and the Korteweg-de Vries (KdV) 
models: in particular, the NLS limit corresponds to the case $\alpha=\sigma=0$, while the complex modified KdV 
(cmKDV) limit corresponds to the case $\rho=\delta=0$. Note that the Hirota equation is a completely integrable model, belonging to 
the Ablowitz-Kaup-Newell-Segur (AKNS) hierarchy \cite{AKNS}. Its derivation in the context of the vortex filament motion is justified in \cite[Sec. 3]{JFM1991}. Briefly speaking, this derivation begins from a natural generalization of the localized induction equation (LIE) which governs the velocity of the vortex filament. The generalization takes account of the axial-flow effect up to the second-order (see \cite[Eqs.~(3.1)-(3.2)]{JFM1991}). Then, the Hirota equation is derived by repeating the original Hasimoto's procedure \cite{Hasi}, which proved the equivalence between LIE and the integrable NLS equation $\mathrm{i}\phi_t+\phi_{xx}+\frac{1}{2}|\phi|^2\phi=0$ (see also \cite{lamb,lambn}). This equivalence implies that LIE is completely integrable. Thus, since the generalized LIE of \cite{JFM1991} preserves integrability, the equivalent evolution equation which is in this case ENLS equation  (\ref{introeq1}), should be also integrable. This is the reason why condition (\ref{cruc2008}) is essential for the association of Eq.~(\ref{introeq1}) with the vortex filament dynamics.
%

On the other hand, non-integrable versions of (\ref{introeq1}), i.e., when the balance condition (\ref{cruc2008}) is 
violated, can be derived in the context of geometric evolution equations. In particular, it was shown in \cite{Onodera1} 
that certain differential equations of one-dimensional dispersive flows into compact Riemann surfaces, may be reduced by a definition of a generalized Hasimoto transform, to ENLS type equations.  See also \cite{Onodera2, Onodera3}, and references therein.

Another important physical context, where certain versions of the ENLS equation 
have important applications, is that of nonlinear optics: such models have been used to describe 
femtosecond pulse propagation in nonlinear optical fibers; in such a case, $t$ and $x$ in Eq.~(\ref{introeq1}) play, 
respectively, the role of the propagation distance along the optical fiber and the retarded time (in a reference frame moving 
with the group velocity), while $\phi(x,t)$ accounts for the (complex) electric field envelope. 
In the context of nonlinear fiber optics, the propagation of short 
optical pulses (of temporal width  
comparable to the wavelength) is accompanied by effects, such as 
higher-order (and, in particular, third-order) linear dispersion, as well as 
nonlinear dispersion (the latter is also called self-steepening term, due to its effect on 
the pulse envelope) \cite{KodHas87}. 
These effects cannot be 
described within the framework of the standard NLS equation,   
but are taken into regard in Eq.~(\ref{introeq1}): 
in the latter, the terms $\propto \phi_{xxx}$ and $|\phi|^2\phi_x$ account, 
respectively, for the third-order (linear) 
dispersion and self-steepening effects. 
Here, it should be noted that while the NLS equation is derived from Maxwell's equations 
in the second-order of approximation in the reductive perturbation method (see, e.g., \cite{KodHas87}), 
the above mentioned higher-order linear and nonlinear 
dispersion effects appear in the third-order of approximation, 
giving rise to ENLS-type equations, \cite{Bindu, Dodd, DJFa1, DJF95a, GromTal2000, KodHas87, DJF95b, Nakk, Potasek87}. 
We also note that a similar situation occurs in the context of nonlinear metamaterials: in such media, propagation 
of ultrashort pulses (in the microwave of optical regimes) is also governed by various versions of ENLS models 
\cite{scalora,tsitsas,wen1,wen2}.

The strong  manifestation of the ENLS (\ref{introeq1}) in many applications, such as the ones mentioned above, motivates our study. This starts from basic analytical results concerning the derivation of conservation laws and global existence of weak solutions of (\ref{introeq1}), proceeds to the investigation of appropriate conditions for the existence of exact soliton solutions (in terms of bright and dark solitons), and finally addresses the modulational instability of plane waves, one of the fundamental mechanisms responsible for the generation of localized nonlinear structures \cite{ZO09}. 


\subsection{Structure of the presentation and main findings}
%
In Section~\ref{SECTION_III}, we focus on the global existence of $H^2$-weak solutions. The result concerns the ENLS (\ref{introeq1}) supplemented with periodic boundary conditions.
The proof is based on the derivation of a differential equation for functionals, up to the level of the $H^2$-norm. We consider this differential equation for brevity, as an $H^2$-``energy equation'' (borrowing the terminology of \cite{JBall2,JBall3} for relevant differential equations on functionals), and stress that it is not 
a conservation law. The main finding is that condition \eqref{cruc2008} is essential for the justification 
of the reduction of higher derivative terms, through a regularization and passage
to the limit procedure leading to its rigorous derivation for weak solutions.  
In this regard, let us recall from the Sobolev embedding theorems, that weakly smooth $H^m(\mathbb{R})$- functions, 
for $m\geq 1$, are $C^{m-1}(\mathbb{R})$-smooth, and that the ENLS equation does not possess any smoothing effect.  Furthermore, we identify an important difference form the NLS limit \cite{XWang}, related to additional implications on the application of the regularized approximation scheme, when the latter should be applied to the ENLS equation: 
in the case of the ENLS, we observe that in the passage to the limit, one has to deal with the appearance and handling
of nonlinear multipliers, instead of the NLS-limit, where only linear multipliers are involved. 
For the global existence of weak solutions, the validity of the conservation of the Hamiltonian corresponding to the NLS-limit for weak solutions under the condition \eqref{cruc2008}, as well as of the $L^2$-norm without any assumption on the coefficients is also crucial. While these conservation laws are known from earlier works \cite{Kim}, 
we revisit  
their derivation noting that the required formal computations are valid at least for solutions initiating from 
$H^3$-smooth initial data. However, we choose to work in the space of weak solutions $H^3$, only 
in order to simplify the presentation of the derivation of the conservation laws: comparing to 
the justification of the $H^2$-energy equation (relevant details are provided in Appendix~\ref{SECTION_II}), 
the revision highlights that the conservation laws can be derived under weaker assumptions on the initial data --see Remark \ref{issue2}. This is a new result showing that conservation laws are valid even in the case of very weak 
assumptions on the smoothness of the initial conditions. 

As a first application of the above results, we apply the 
$H^2$-energy equation
to prove global existence of the weak solutions at the level of that norm, by combining it
with appropriate alternative theorems \cite{Alb}. It is interesting to recover that in addition to the balance condition (\ref{cruc2008}), the presence of the second-order dispersion $\rho\neq 0$ is required, as a necessary condition for global existence.
We note that the periodic case is in general quite different from the Cauchy
case since dispersive effects are not expected (cf.
\cite{BourgainAMS}), and can lead even to blow-up phenomena in
finite time for such types of equations, despite the fact that they
may contain just cubic nonlinear terms (see \cite{HorikiriDJF} for a ENLS-type equation, and \cite{Zhang08}  on the periodic
Dullin-Gottwald-Holm equation). Furthermore, our lines of approach are justified by the important outcomes of Kato's theory \cite{Kato0}-\cite{Kato3}, that the blow-up occurs in any $H^k$, $k\geq2$-norm, if it occurs at all; see also \cite{cazS}. 
The global well-posedness of the Cauchy problem for ENLS type equations in $H^k(\mathbb{R})$,  
$k\geq 2$ has been addressed  in \cite{Laurey97},
with most recent results those of \cite{Guo05,Segata08} for the global well
posedness for $1\leq k\leq 2$.
We remark that the importance of the rigorous justification of conservation laws and energy
equations has been also underlined in \cite{JBall1, JBall2, JBall3}, and in \cite{Ozawa2005} which considers the conservation laws for the Cauchy problem of the higher dimensional NLS
equation. 

In Section \ref{SECTION_IV}, we investigate conditions under which the ENLS (\ref{introeq1}) 
may support traveling wave solutions. The analysis is based on the reduction of (\ref{introeq1}) to 
an ordinary differential equation (ODE), by means of a traveling-wave ansatz. 
This method has been extensively used for the derivation of various types of traveling waves 
-- such as periodic \cite{Feng,DJF95b}, solitary \cite{Kl1, DJF95b, Kl2} and unbounded \cite{DJF95b, Kl2} ones -- 
of the Hirota and other versions of the higher-other NLS model. 
%
We revisit  
the phase-plane analysis of the reduced dynamical system 
to obtain conditions for the existence of  
bright and dark solitons, as well as  
the solutions themselves. 
The main findings of this section (see also Remark \ref{issue4}) are the following. 
First, traveling solitons exist only under the balance condition (\ref{cruc2008}). 
Second, we reveal the role of the wavenumber $k_0$ of the carrier wave: if $k_0 \ne 0$ and $k_0 \ne \rho/3\sigma$ 
traveling wave solutions exist only in the presence of both second ($\rho\neq 0$) and third-order ($\sigma\neq 0$) 
dispersion; if $k_0=0$ traveling waves may exist in both cases $\rho \ne 0$ and $\rho =0$, corresponding to the so-called 
``zero-dispersion point'' (ZDP) in the context of nonlinear fiber optics \cite{CPag95}. Note that in the ZDP case, 
we find that the carrier frequency should also vanish, in which case the balance condition (\ref{cruc2008}) yields 
$\delta=0$, and the ENLS~(\ref{introeq1}) becomes the cmKdV model.
This way, we also elucidate the importance of the linear dispersion effect on the 
existence of soliton solutions. 
Third, yet another interesting finding is that standing wave solutions are not supported by the ENLS (\ref{introeq1}).

%

We also present results of direct numerical simulations 
illustrating that, when \eqref{cruc2008} is satisfied, the solitons are robust 
under small-amplitude random-noise perturbations. 
On the other hand, when \eqref{cruc2008} is violated, our numerical 
results illustrate that an 
arbitrary initial solitary pulse is always followed by emission of radiation (which tends to disperse at later times.  
The results on the robustness of the solitons, as well as on the radiation emission when 
\eqref{cruc2008} is violated, are in accordance to the findings of Ref.~\cite{GromTal2000} 
and the analysis of \cite{YangPeli2, Yang1}, which refer to the case of the bright solitons. 
However, the numerical simulations also reveal important differences, regarding the dynamics of bright and dark pulses. 
For instance, regarding the emission of radiation, we find that a dark solitary wave 
initial condition tends to a single dark pulse, which is eventually 
separated by 
the emitted radiation; the latter, is traveling faster from the soliton, and assumes the
form of an apparent dispersive shock wave~\cite{DSW}.
On the contrary, in the case of a bright solitary wave the radiation arises
behind the principal pulse. Another 
important difference concerns the direction of the propagation of the pulses: while 
the evolution of bright solitons 
is associated with a turning effect -- i.e., the soliton initially starts moving towards the positive $x$-direction, but eventually travels with an almost constant velocity towards the negative $x$-direction -- this effect is not observed 
in the evolution of the dark soliton. This effect for the bright solitary wave 
structures is somewhat reminiscent of direction-reversing solitary waves, so-called ``boomerons''; see, e.g.,  Ref.~\cite{degasp} and references therein.
%
%
We note that all numerical experiments have been performed by using a  pseudo-spectral method, under which the numerical quantities related to the conservation laws of the ENLS equation under the balance condition (\ref{cruc2008}), are also conserved at a high order of accuracy.

In Section \ref{SECTION_V}, we perform a modulation instability (MI) analysis of plane wave solutions of 
(\ref{introeq1}).  Here, we should recall that MI is the process by which a plane wave of the form 
$\phi(x,t)=\phi_0\mathrm{e}^{\mathrm{i}(kx-\omega t)}$, $\phi_0\in\mathbb{R}$, becomes unstable and gives rise 
to the emergence of nonlinear structures \cite{ZO09}. 
One of the main findings of the MI analysis, is the detection of the drastic differences between the conditions for MI on the ENLS equation (\ref{introeq1}) and its NLS-limit; in particular, we find the following.
If the wavenumber of the plane wave solution is $k=0$ then the MI conditions for the focusing counterpart of the ENLS equation (\ref{eq3.18}) are exactly the same with its focusing NLS-limit. Additionally, the plane waves are always modulationally stable in the defocusing ENLS equation, as it happens in its defocusing NLS-limit. 
On the other hand, when $k\neq 0$, we reveal MI conditions for both the focusing and the defocusing 
counterparts of the ENLS equation that are clearly
{\it distinct} from their NLS counterparts. Specifically, the MI conditions  
in the focusing case are similar 
to those of the focusing NLS-limit but the instability band may decrease or increase depending on the values of the parameters involved;
on the other hand, we find that in the defocusing case of the ENLS equation, MI may be manifested. 
This is a drastic difference with the defocusing NLS limit where the plane waves are always stable. 
The analytical considerations of the MI analysis 
are concluded by the justification that, in the MI regime, the Fourier modes of the unstable plane wave solution 
are the integer multiples of the unstable wave number of the perturbation
(the well known ``sidebands'' pertaining to the MI). 
The results of the numerical studies are in excellent agreement with the analytical predictions of the MI analysis.

Finally, Section  \ref{SECTION_VI} offers a discussion and a summary of our results, and discusses some future directions.


%
\section{Conservation laws and global existence of
solutions}
\label{SECTION_III}
%
%
In this section, we consider the
derivation of conservation laws and energy-type 
equations for the ENLS equation (\ref{introeq1}). These quantities will have a crucial role in proving global existence of weak solutions,
for the periodic initial-boundary value problem of this equation.   More precisely, we supplement the ENLS equation (\ref{introeq1})  with space-periodic boundary conditions, 
\begin{equation}\label{bc}
\phi(x+L,t)=\phi(x,t),
\quad\forall\,x\in{\mathbb R},\;\;\forall\,t\in{\mathbb R},
\end{equation}
where $L>0$  is given, and with the initial condition
\begin{equation}\label{introeq2}
\phi (x,0)=\phi_0(x),\;\;\forall\,x\in{\mathbb R}.
\end{equation}
%
%
%
%
\par Let us recall some preliminary information on the functional setting of the problem, as well as \textit{a local existence result}. 
Problem (\ref{introeq1})-(\ref{bc})-(\ref{introeq2}) will be considered in
complexifications of the  Sobolev spaces $H^m_{per}(\Omega)$ of
{\em real} periodic functions, locally in $H^m(\mathbb{R})$, where
$m$ is a nonnegative integer and $\Omega=(0,L)$. These spaces
$H^m_{per}(\Omega)$ are endowed with the scalar product and the
induced from it norm, given below
\begin{equation}\label{pernorm}
(\phi,\psi)_{\ssy
H^m(\Omega)}:=\sum_{j=0}^m
\int_{\ssy\Omega}\partial^j\phi\,\partial^j\overline{\psi}
\;dx
\quad\text{\rm and}\quad
\|\phi\|_{\ssy H^m(\Omega)}:=\sqrt{(\phi,\phi)_m},
\end{equation}
where $\partial^j:=\frac{\partial^j}{\partial x^j}$.
The case $m=0$, simply corresponds to $L^2(\Omega)$. The spaces
$H^m_{per}(\Omega)$ can be studied by using Fourier series
expansion, and they can be characterized as
\begin{equation*}
H^m_{per}(\Omega)=\left\{\sum_{k\in\mathbb{Z}}\widehat{\phi}_k\,\mathrm{e}^{{\rm
2\pi i}\,\frac{k}{L}\,x}:
\quad\widehat{\phi}_k=\overline{{\widehat\phi}_{-k}}\quad
\forall\,k\in{\mathbb Z},
\quad\sum_{k\in\mathbb{Z}}(1+|k|^2)^m\,|\widehat{\phi}_k|^2
<\infty\right\},
\end{equation*}
where $\widehat{\phi}_k=\overline{\widehat{\phi}_{-k}}$ are the
(real) Fourier coefficients~\cite{RTem2}. The scalar product and the norm in
$H^m_{per}(\Omega)$ given by \eqref{pernorm}, are equivalent to
those defined in terms of Fourier coefficients
\begin{eqnarray}
\label{perf}
(\phi,\psi)_{m,*}=\sum_{k\in\mathbb{Z}}(1+|k|^2)^m\,\widehat{\phi}_k
\,\widehat{\psi}_k
\quad\text{\rm and}\quad
\|\phi\|_{m,*}=\left[\,\sum_{k\in\mathbb{Z}}(1+|k|^2)^m\,
|\widehat{\phi}_k|^2\right]^{\frac{1}{2}},
\end{eqnarray}
i.e., there exist constants $c_1,c_2>0$ such that
\begin{equation}
\label{eqSn}
c_1\,\|\phi\|_{\ssy H^m(\Omega)}\,\leq\,\|\phi\|_{m,*}
\,\leq\,c_2\,\|\phi\|_{\ssy H^m(\Omega)}.
\end{equation}
The complexified space of $H^m_{per}(\Omega)$ becomes a {\em real}
Hilbert space, (which shall be still denoted for convenience by
$H^m_{per}(\Omega)$), if it is endowed with the inner product and
norm
%
\begin{equation}\label{complf}
\begin{gathered}
{[}\phi,\psi]_{\ssy H^m(\Omega)}:= \mathrm{Re}
\left[\,\sum_{j=0}^m \int_{\ssy\Omega}\partial^j\phi
\,\partial^j\overline{\psi}\;dx\,\right]
=(\phi_1,\psi_1)_{\ssy H^m(\Omega)}+(\phi_2,\psi_2)_{\ssy H^m(\Omega)},\\
\|\phi\|_{\ssy H^m(\Omega)}:=\sqrt{[\phi,\phi]_{\ssy
H^m(\Omega)}},
\end{gathered}
\end{equation}
where $\phi=\phi_1+{\rm i}\,\phi_2$ and $\psi=\psi_1+{\rm
i}\,\psi_2$. The equivalent norms in terms of Fourier coefficients
are defined, taking into account \eqref{perf} and \eqref{eqSn}, in
an analogous manner to \eqref{complf}.
\par
For $m\geq 1$, the following embedding relations
\begin{equation}\label{prop1}
H^m_{per}(\Omega)\subseteq C^{m-1}(\Omega),
\quad
H^m_{per}(\Omega)\subseteq H^{m-1}_{per}(\Omega),
\end{equation}
are compact. Moreover $H^m_{per}(\Omega)$ is a generalized Banach
algebra: there exists a constant $c$ depending only on $\Omega$,
such that if $z,\,\eta\in H^m_{per}(\Omega)$, then $z\cdot\eta\in
H^m_{per}(\Omega)$ and
\begin{equation}\label{prop2}
\|z\cdot\eta\|_{\ssy H^m(\Omega)}\leq\,c\,\|z\|_{\ssy
H^m(\Omega)}\,\|\eta\|_{\ssy H^m(\Omega)}.
\end{equation}
We also recall the Gagliardo-Nirenberg inequality for the one
space dimension case.
Let $\Omega\subseteq{\mathbb R}$, $1\leq p,q,r\leq\infty$, $j$  an
integer, $0\leq j\leq m$ and $\frac{j}{m}\leq\theta\leq 1$. Then
\begin{equation}\label{gnin}
\|\partial^j\phi\|_{\ssy L^p(\Omega)}\leq\,C\,
\|\phi\|^{1-\theta}_{\ssy L^q(\Omega)}\,\left(\,\sum^m_{k=0}
\|\partial^k\phi\|_{\ssy L^r(\Omega)}\,\right)^{\theta},
\quad\forall\,\phi\in L^q(\Omega)\cap W^{m,r}(\Omega),
\end{equation}
where
$\frac{1}{p}=j+\theta\left(\frac{1}{r}-m\right)+\frac{1-\theta}{q}$.
If $m-j-\frac{1}{r}=0$ and $r>1$, then the
inequality holds for $\theta <1$.
\par
One of the methods which can be used in order to prove a local existence result, is the method of  \textit{semibounded evolution equations} which was developed in \cite{Katoat}. The application of this method to \eqref{in1}-\eqref{introeq2}-\eqref{bc} although involving lengthy computations, is now considered as standard. Thus, we refrain to show the details here, and we just state the local existence result in
%
\begin{thm}\label{thmloc}
Let $\phi_0\in H^k_{per}(\Omega)$ for any integer $k\geq2$, and $\alpha, \rho,\sigma,\delta \in\mathbb{R}$.
Then there exists $T>0$, such that the problem
\eqref{introeq1}-\eqref{bc}-\eqref{introeq2}, has a unique
solution 
\begin{equation*}
\phi\in C([0,T], H^k_{per}(\Omega)) \quad\mbox{and}\quad\phi_t\in
C([0,T], H^{k-3}_{per}(\Omega)).
\end{equation*}
\end{thm}
%
%
\setcounter{remm}{1}
\begin{remark}\label{issue1}
(Generalized equation). {\em Except for the method of \cite{Kato2,Katoat}, a modification of the methods for local existence of solutions for  nonlinear hyperbolic systems is applicable to the generalized ENLS equation,
\begin{equation}\label{geneq1}
\phi_t+\beta(|\phi|^2)\,\phi_x-{\rm
i}\,\rho\,\phi_{xx}+\sigma\,\phi_{xxx} -{\rm
i}\,\gamma(|\phi|^2)\,\phi=0\quad\text{\rm on}\quad{\mathbb
R}\times(0,T],
\end{equation}
where the given functions
\begin{equation}\label{regul}
\beta,\gamma\in C^1({\mathbb R};{\mathbb R}).
\end{equation}
For instance, it can be proved in a similar manner (see also \cite[Theorem 1.2, pg. 362 \& Proposition 1.3, pg. 364]{TaylorII}), that if $\phi_0\in H^{k}_{per}(\Omega)$, $k\geq 2$, then $\phi\in C([0,T], H^k_{per}(\Omega))$. Note that the problem does not possess, a-priori, any smoothing effect. Moreover, the arguments of \cite{Kato2, Katoat, TaylorII}, imply continuous dependence on initial data in $H^k_{per}(\Omega)$.}
\end{remark}
%
%
It is known (see, e.g., Ref.~\cite{Kim}), that  solutions of the problem \eqref{introeq1}-\eqref{introeq2}-\eqref{bc} 
possess conservation laws similar to the conservation of the $L^2$-norm and energy as the 
the solutions of its NLS-limit, i.e. the quantities
\begin{gather}
\mathbf{N}(\phi):=\int_{\ssy\Omega}|\phi|^2\;dx
=\|\phi\|^2_{\ssy L^2(\Omega)},\label{law1}\\
\mathbf{J}(\phi):=\tfrac{\rho}{2}\,\int_{\ssy\Omega}|\phi_x|^2\;dx
-\tfrac{\delta}{4}\,\int_{\ssy\Omega}|\phi|^4\;dx.\label{law2}
\end{gather}
are conserved. Notice that while the first $L^2$ norm is conserved
for {\it all} parameter combinations, the second conservation law
only holds if  Eq.~(\ref{cruc2008}) is satisfied. 
Here, first we briefly discuss their derivation, in the case of sufficiently smooth initial data, 
with the regularity suggested by Theorem 2.1 for $k=3$; nevertheless, below we extend our considerations 
and show that the conservation laws are satisfied by weak solutions, even when the latter are 
initiating for considerably weak initial data, belonging to an appropriate class of Sobolev spaces 
(cf. Remark 2.6 below).
%

Hence, assuming $H^3$-smoothness of the initial data, we recall the following
%
%
\setcounter{lem}{2}
\begin{lem}\label{conslaws}
For $\phi_0\in H^3_{per}(\Omega )$, let $T^*$ be the maximum time
for which for all $0<T<T^*$, the solution $\phi$ of
\eqref{introeq1}-\eqref{bc}-\eqref{introeq2} lies in
$C([0,T],H^3_{per}(\Omega))\cap C^1([0,T], L^2(\Omega))$. 
Then for any $t\in [0,T^*)$, the solution satisfies the conservation law
\begin{eqnarray}
\mathbf{N}(\phi(t))=\mathbf{N}(\phi_0),\label{consch}
\end{eqnarray}
for all $\alpha,\rho,\sigma,\delta\in\mathbb{R}$. Furthermore, if (\ref{cruc2008}) is satisfied, 
then the NLS Hamiltonian is conserved, that is, 
\begin{eqnarray} 
\mathbf{J}(\phi(t))=\mathbf{J}(\phi_0).\label{coneneg}
\end{eqnarray}
\end{lem}
%
%
%
%
%
\begin{proof}
For convenience, we rewrite equation (\ref{introeq1}) as
\begin{equation}\label{in1}
\phi_t-{\rm i}\,\rho\,\phi_{xx}+\sigma\,\phi_{xxx}
=-3\alpha\,|\phi|^2\phi_x+{\rm i}\,\delta \,|\phi|^2\phi.
\end{equation}
Since $\phi\in C([0,T],H^3_{per}(\Omega))\cap C^1([0,T],
L^2(\Omega))$, it makes sense to multiply equation \eqref{in1} by
$\overline{\phi}$, integrate with respect to space and keep real
parts. We obtain the equation
\begin{equation}\label{coa}
\tfrac{1}{2}\,\tfrac{d}{dt}\|\phi(t)\|^2_{\ssy L^2(\Omega)}
+\tfrac{3\alpha}{2}\,\mathrm{Re}\left[\,\int_{\ssy\Omega}
|\phi|^2(|\phi|^2)_x\;dx\,\right]
+\sigma\mathrm{Re}\left[\,\int_{\ssy\Omega}
\phi_{xxx}\overline{\phi}\;dx\right]=0.
\end{equation}
By periodicity, the last two integrals vanish, and \eqref{consch}
follows. 

To derive the conservation law \eqref{coneneg}, we
multiply \eqref{in1} by $\overline{\phi}_t$, integrate in space,
keeping this time the imaginary parts. Now, the resulting equation is
\begin{equation}\label{coen1}
\tfrac{d}{dt}J(\phi(t))+3\,\alpha\,\mathrm{Im}\left[
\,\int_{\ssy\Omega}|\phi|^2\,\phi_x\overline{\phi}_t\;dx\,\right]
+\sigma\,\mathrm{Im}\left[\int_{\ssy\Omega}\phi_{xxx}
\,\overline{\phi}_t\;dx\,\right]=0.
\end{equation}
By substitution of $\overline{\phi}_t=-\mathrm{i}\rho\overline{\phi}_{xx}-\sigma\overline{\phi}_{xxx}-3\alpha|\phi|^2\overline{\phi}_x-\mathrm{i}\delta |\phi|^2\overline{\phi}$, into the two integral terms of the left-hand side of (\ref{coen1}), and integration by parts, we observe that these two terms become:
\begin{equation}\label{coen2}
\begin{split}
3\,\alpha\,\mathrm{Im}\left[\int_{\ssy\Omega}|\phi|^2
\,\phi_x\,\overline{\phi}_t\;dx\right]
=&3\,\alpha\,\mathrm{Im}\left[\int_{\ssy\Omega}|\phi|^2\,\phi_x
\left(-3\,\alpha\,|\phi|^2\,\overline{\phi}_x-{\rm
i}\,\rho\,\overline{\phi}_{xx} -\sigma\,\overline{\phi}_{xxx}
-{\rm i}\,\delta\,|\phi|^2\,\overline{\phi}\right)\;dx\right]\\
=&-3\,\alpha\,\rho\mathrm{Re}\left[\int_{\ssy\Omega}|\phi|^2
\,\phi_x\,\overline{\phi}_{xx}\;dx\right]
-3\,\alpha\,\sigma\,\mathrm{Im}\left[\int_{\ssy\Omega}|\phi|^2
\,\phi_x\,\overline{\phi}_{xxx}\;dx\right],\\
&-\tfrac{3\alpha\rho}{2}\,\left[\int_{\ssy\Omega}
|\phi|^2\,(|\phi_x|^2)_x\;dx\right]
+3\,\alpha\,\sigma\,\mathrm{Im}\left[\int_{\ssy\Omega}|\phi|^2
\,\overline{\phi}_x\,\phi_{xxx}\;dx\right],
\end{split}
\end{equation}
and
\begin{equation}\label{coen3}
\sigma\,\mathrm{Im}\left[\int_{\ssy\Omega}\phi_{xxx}
\,\overline{\phi}_t\;dx\right]=
-\sigma\,\delta\,\mathrm{Re}\left[\int_{\ssy\Omega}|\phi|^2
\,\phi_{xxx}\,\overline{\phi}\;dx\right]
-3\,\alpha\,\sigma\,\mathrm{Im}\left[\int_{\ssy\Omega}|\phi|^2
\,\overline{\phi}_x\,\phi_{xxx}\;dx\right].
\end{equation}
Then, by using \eqref{coen2} and \eqref{coen3}, equation \eqref{coen1} can be
written as
\begin{equation}\label{coen4}
\tfrac{d}{dt}J(\phi(t))
-\tfrac{3\alpha\rho}{2}\,\left[\int_{\ssy\Omega}
|\phi|^2\,(|\phi_x|^2)_x\;dx\right]
-\sigma\,\delta\,\mathrm{Re}\left[\int_{\ssy\Omega}
|\phi|^2\,\phi_{xxx}\,\overline{\phi}\;dx\right]=0.
\end{equation}
For the last term of the left-hand side of (\ref{coen4}), we observe that
\begin{gather}
-\sigma\,\delta\,\mathrm{Re}\left[\int_{\ssy\Omega}|\phi|^2
\,\phi_{xxx}\,\overline{\phi}\;dx\right]
=\tfrac{\sigma\delta}{2}\,\left[\int_{\ssy\Omega}
|\phi|^2\,(|\phi_x|^2)_x\;dx\right]
+\sigma\,\delta\,\mathrm{Re}\left[\int_{\ssy\Omega}
\phi_{xx}\,\overline{\phi}\,(|\phi|^2)_x\;dx\right],\label{coen5}
\end{gather}
while for the last term of the right-hand side of (\ref{coen5}),
\begin{gather}
\sigma\,\delta\,\mathrm{Re}\left[\int_{\ssy\Omega}
\phi_{xx}\,\overline{\phi}\,(|\phi|^2)_x\;dx\right]=
\tfrac{\sigma\delta}{2}\,\left[\int_{\ssy\Omega}
|\phi|^2\,(|\phi_x|^2)_x\;dx\right]
+\sigma\,\delta\,\mathrm{Re}\left[\int_{\ssy\Omega}
\phi_{xx}\,\phi_x\,\overline{\phi}^2\;dx\right].\label{coen6}
\end{gather}
Also, for the last term of the right hand side of \eqref{coen6}, we have
\begin{equation*}
\begin{split}
\sigma\,\delta\,\mathrm{Re}\left[\int_{\ssy\Omega}
\phi_{xx}\,\phi_x\,\overline{\phi}^2\;dx\right]
=&-\,\sigma\,\delta\,\mathrm{Re}\left[\int_{\ssy\Omega}
\phi_x\,\left(\phi_{xx}\,\overline{\phi}^2
+2\,\phi_x\overline{\phi}\,\overline{\phi}_x\right)\;dx\right]\\
=&-\,\sigma\,\delta\,\mathrm{Re}\left[\int_{\ssy\Omega}
\phi_{xx}\,\phi_x\,\overline{\phi}^2\,dx\right]
-2\,\sigma\,\delta\,\mathrm{Re}\left[
\int_{\ssy\Omega}|\phi_x|^2\,\phi_x\,\overline{\phi}\,dx\right].
\end{split}
\end{equation*}
Hence, we have for this term, that
\begin{equation}\label{coen7}
\begin{split}
\sigma\,\delta\,\mathrm{Re}\left[\int_{\ssy\Omega}
\phi_{xx}\,\phi_x\,\overline{\phi}^2\;dx\right]=&-\,\sigma\,\delta\,\mathrm{Re}\left[
\int_{\ssy\Omega}|\phi_x|^2\,\phi_x\,\overline{\phi}\,dx\right]\\
=&-\tfrac{\sigma\delta}{2}\,\left[
\int_{\ssy\Omega}(|\phi|^2)_x\,|\phi_x|^2\;dx\right]=
\tfrac{\sigma\delta}{2}\,\left[
\int_{\ssy\Omega}|\phi|^2\,(|\phi_x|^2)_x\;dx\right].
\end{split}
\end{equation}
Eventually, it follows from \eqref{coen5}-\eqref{coen7}, that the last term of (\ref{coen4}) is
\begin{equation}\label{coen8}
-\sigma\,\delta\,\mathrm{Re}\left[\int_{\ssy\Omega}
|\phi|^2\,\phi_{xxx}\,\overline{\phi}\;dx\right]
=\tfrac{3\delta\sigma}{2}\,\mathrm{Re}\left[
\int_{\ssy\Omega}|\phi|^2\,(|\phi_x|^2)_x\;dx\right].
\end{equation}
Therefore, equation (\ref{coen4}) is taking the form
\begin{equation}\label{coen4FIN}
\tfrac{d}{dt}J(\phi(t))=
\tfrac{3}{2}(\alpha\rho-\sigma\delta)\,\left[\int_{\ssy\Omega}
|\phi|^2\,(|\phi_x|^2)_x\;dx\right]=-\tfrac{3}{2}(\alpha\rho-\sigma\delta)\,\left[\int_{\ssy\Omega}
(|\phi|^2)_x\,|\phi_x|^2\;dx\right].
\end{equation}
The assumption (\ref{cruc2008}) on the coefficients, shall be used on the equation 
\eqref{coen4FIN}, implies that $\frac{d}{dt}\,J(\phi (t))=0$, i.e., the conservation law (\ref{coneneg}).
\end{proof}
%
%
%
\par
The conservation laws provided by Lemma~\ref{conslaws} will be
useful in order to establish global in time regularity for the
solutions of \eqref{introeq1}-\eqref{introeq2}-\eqref{bc}, at
$H^2_{per}(\Omega)$-level. For this purpose, we derive 
a differential equation for 
functionals involving 2nd-order derivatives. For brevity, we call this differential equation as ``energy equation'' 
(as per the terminology used in Refs.~\cite{JBall2,JBall3}), but stressing that it is not a conservation law.

%
\setcounter{lem}{3}
\begin{lem}
\label{weaksol1}
For $\phi_0\in H^3_{per}(\Omega )$, let $T^*$ be the maximum time
for which for all $0<T<T^*$, the solution $\phi$ of
\eqref{introeq1}-\eqref{bc}-\eqref{introeq2}, lies in
$C([0,T],H^3_{per}(\Omega))\cap C^1([0,T], L^2(\Omega))$. Moreover,
assume that the coefficients satisfy \eqref{cruc2008}. We consider the
functionals
\begin{equation}\label{pseudocons}
\mathbf{J_1}(\phi):=\rho\,\int_{\ssy\Omega}|\phi_{xx}|^2\;dx
-4\,\delta\,\left[\int_{\ssy\Omega}
|\phi|^2\,|\phi_x|^2\;dx\right]
-\delta\,\mathrm{Re}\left[\int_{\ssy\Omega}
\phi^2\,\overline{\phi}^2_x\;dx\right],
\end{equation}
and
\begin{equation}\label{finmultip}
\begin{split}
\mathbf{E_1}(t):=&\,9\,\alpha\,\delta\,
\left[\int_{\ssy\Omega}|\phi|^2\,|\phi_x|^2\,(|\phi|^2)_x\;dx\right]
-6\,\alpha\,\delta\,\mathrm{Re}\left[\int_{\ssy\Omega}|\phi|^2
\,\phi_x\,\overline{\phi}^2\,\phi_{xx}\;dx\right]\\
&\,-8\,\delta\,\rho\,\mathrm{Im}\left[\int_{\ssy\Omega}
|\phi|^2\,\phi\,\overline{\phi}_{xx}\;dx\right]
-4\,\delta\,\rho\,\mathrm{Im}\left[\int_{\ssy\Omega}
|\phi_x|^2\,\overline{\phi}\,\phi_{xx}\;dx\right]\\
&\,+6\,\delta^2\,\mathrm{Im}\left[\int_{\ssy\Omega}|\phi|^4
\,\overline{\phi}\,\phi_{xx}\;dx\right]
-6\,\delta^2\,\int_{\ssy\Omega}|\phi|^2\,\phi^2\overline{\phi}_x\;dx.\\
\end{split}
\end{equation}
Then the solution $\phi$, for any $t\in [0,T]$ satisfies the
energy equation
\begin{equation}\label{hirotaenergy}
\tfrac{d}{dt}\,\mathbf{J_1}(\phi)=\mathbf{E_1}(t).
\end{equation}
\end{lem}
%
%
%
%
%
\begin{proof}
An attempt for the formal derivation of (\ref{hirotaenergy}) readily reveals the generation of higher order derivative terms, on which, direct calculations are not permitted by the weak regularity assumptions on the initial data. Therefore, 
we shall rigorously calculate the time-derivative of the functional (\ref{pseudocons})  (as suggested by the right-hand side of the claimed equation (\ref{hirotaenergy})), by following an
approximation scheme (close to that developed in \cite{XWang} for the linearly damped and forced NLS equation).  The scheme is based on a
smooth approximation of the solution of
\eqref{introeq1}-\eqref{introeq2}-\eqref{bc}, both in space and time. 
We denote by
\begin{equation}\label{defOm}
\ld\phi,\psi\rd_{H^m,H^{-m}}=\mathrm{Re}\left[
\int_{\ssy\Omega}\phi\,\psi\;dx\right],
\end{equation}
the duality bracket between the space $H^m$ and its dual $H^{-m}$,
$m\geq 1$, to elucidate that several quantities act as
functionals during our calculations. Note also that due to the (compact)
embeddings $H^m_{per}(\Omega)\subset L^2(\Omega)\subset
H^{-m}_{per}(\Omega)$, $m\geq 1$, the identification
\begin{equation}\label{Sembe}
\ld\phi,\psi\rd_{H^m,H^{-m}}=(\phi,\psi)_{\ssy L^2(\Omega)},
\quad\mbox{for every}\quad\phi\in H^m_{per}(\Omega),\quad\psi\in
L^2(\Omega),
\end{equation}
holds. 

For the smooth approximation in space,  we consider for
arbitrary $n\in\mathbb{N}$, the low frequency part of the unique solution $\phi\in
C([0,T],H^3_{per}(\Omega))\cap C^1([0,T], L^2(\Omega))$ of
\eqref{introeq1}-\eqref{introeq2}-\eqref{bc},
\begin{equation}\label{Tem1}
\phi_n(t)=\sum_{|k|\leq n}\widehat{\phi}_k(t)\,\mathrm{e}^{{\rm
2\pi i}\,\frac{k}{L}\,x},
\end{equation}
which is the projection of the solution $\phi$ onto the first $n$
Fourier modes. The function $\phi_n$ is a smooth analytic
function with respect to the space variable.  A regularization of
$\phi_n$ in time, can be obtained by a standard mollification in
time. For instance, we consider the standard mollifier in time
$J_{\epsilon}(t)=\epsilon^{-1}J(t/\epsilon)$, $\epsilon>0$, with
the properties $J_{\epsilon}(t)=0$ if $|t|\geq\epsilon$, and
$\int_{\ssy\mathbb{R}}J_{\epsilon}(t)dt=1$. Recall that $J\in
C_0^{\infty}(\mathbb{R})$, $J\geq 0$ on $\mathbb{R}$, and
$\int_{\ssy\mathbb{R}}J(x)dx=1$. Then, the function
\begin{equation*}
J_{\epsilon}*\phi_n\,(t)=\int_{\ssy\mathbb{R}}J_{\epsilon}(t-s)\phi_n(s)ds,
\end{equation*}
is infinitely differentiable function on $\mathbb{R}$ into
$H^3_{per}(\Omega)$. Furthermore, by setting
$\epsilon=\frac{1}{n}$, we see that the regularized
sequence $v_n=J_{\frac{1}{n}}*\phi_n$, satisfies the key convergence relations
to $\phi\in L^2([0,T],X)$, $X=H^m_{per}(\Omega)$, $-2\leq m\leq 3$, given in (\ref{propvn}) --cf. Appendix~\ref{SECTION_II}.

We indicatively demonstrate the calculations required for the time derivative $\tfrac{d}{dt}\left( |\phi_{x}|^2,|\phi|^2\right)_{\ssy
L^2(\Omega)}$-the time derivative of the second term of of the functional (\ref{pseudocons}).
We note again that more details can be found in Appendix~\ref{SECTION_II}. By using the approximating
sequence $v_n$,  Lemmas \ref{xwa1} and \ref{xwa2}, and performing the calculations for $v_n$ in the sense of the dual spaces as in (\ref{multip1}), we derive the equation
\begin{equation}\label{prepassage}
\begin{split}
\left( |v_{nx}(t)|^2,|v_n(t)|^2\right)_{\ssy L^2(\Omega)} -\left(
|v_{nx}(0)|^2,|v_n(0)|^2\right)_{\ssy L^2(\Omega)}
=&-2\int_{0}^{t}\ld |v_n|^2\overline{v}_{nxx},v_{nt}
\rd_{\ssy H^1,H^{-1}}\;d\tau\\
&-2\int_0^{t}\ld \overline{v}_{nx}^2 v_n,v_{nt}\rd_{\ssy
H^2,H^{-2}}d\tau,
\end{split}
\end{equation}
which holds for every $t\in [0,T]$. A difference from the NLS 
equation \cite{XWang}, is that in the passage to
the limit to (\ref{prepassage}) (or its equivalent \eqref{multip1}), we have to deal with the appearance
of nonlinear multipliers in the right-hand side. We first observe
that by using \eqref{propvn} and Aubin-Lions \cite[Corollary 4,
p.85]{sim90}, we may derive the strong convergence
\begin{equation}
\begin{array}{lllll}
\label{multip1a}
v_{n}&\rightarrow& \phi,&\mbox{in}&L^2([0,T], H^{2}_{per}(\Omega)),\\
v_{nxx}&\rightarrow& \phi_{txx},&\mbox{in}&L^2([0,T], L^2(\Omega)).
\end{array}
\end{equation}
On the other hand, we may see that there exists some $c>0$ such
that
\begin{equation*}
\int_{0}^{\ssy
T}\|\,|v_n|^2{v}_{nxx}-|\phi|^2{\phi}_{xx}\,\|_{\ssy
L^2(\Omega)}^2\,ds\leq\, c\,\left\{\int_{0}^{\ssy
T}\|{v}_{nxx}-\phi_{xx}\|^2_{\ssy L^2(\Omega)}\;ds+\int_{0}^{\ssy
T}\|{v}_{n}-\phi\|^2_{\ssy L^2(\Omega)}\;ds\right\}\rightarrow 0,
\end{equation*}
as $n\rightarrow\infty$; the second term in the right hand
side of \eqref{prepassage} may be treated similarly.  Letting
$n\rightarrow\infty$,  equation \eqref{prepassage} (or \eqref{multip1}) converges to
\begin{equation}\label{multip1b}
\tfrac{d}{dt}\left( |\phi_{x}|^2,|\phi|^2\right)_{\ssy
L^2(\Omega)}=-2\ld |\phi|^2\overline{\phi}_{xx},\phi_{t}\rd_{\ssy
H^1,H^{-1}}-2\ld \overline{\phi}_{x}^2 \phi,\phi_{t}\rd_{\ssy
H^2,H^{-2}}.
\end{equation}
From \eqref{multip1b}, we may proceed by inserting the expression for $\phi_t$ provided by
\eqref{in1}. Manipulating the  right-hand side of \eqref{multip1b} as shown in (\ref{multip1c}), we derive the equation (\ref{multip2}), which is valid for the solution $\phi\in C([0,T],H^3_{per}(\Omega))$, since it  only contains (weak) derivatives up to the third order. This fact allows for further integration by parts to the right-hand side of (\ref{multip2}).  Eventually, we arrive in 
\begin{equation}\label{multip5}
\begin{split}
\tfrac{d}{dt}\int_{\ssy\Omega}|\phi|^2\,|\phi_x|^2\;dx
=&\,3\,\alpha\,\mathrm{Re}\left[\int_{\ssy\Omega}|\phi|^4\,(|\phi_x|^2)_x\;dx\right]
+3\,\sigma\,\mathrm{Re}\left[\int_{\ssy\Omega}|\phi|^2\,(|\phi_{xx}|^2)_x\;dx\right]\\
&+3\,\alpha\,\mathrm{Re}\left[\int_{\ssy\Omega}|\phi|^2\,|\phi_x|^2
\,(|\phi|^2)_x\;dx\right]
+2\,\delta\,\mathrm{Im}\left[\int_{\ssy\Omega}|\phi|^4\,\phi
\,\overline{\phi}_{xx}\;dx\right]\\
&+2\,\rho\,\mathrm{Im}\left[\int_{\ssy\Omega}
\overline{\phi}_x^2\,\phi\,\phi_{xx}\;dx\right]
+2\,\delta\,\mathrm{Im}\left[
\int_{\ssy\Omega}\overline{\phi}_x^2\,\phi^2\,|\phi|^2\;dx\right].\\
\end{split}
\end{equation}
For the time derivatives of the rest of the terms of (\ref{pseudocons}), we work similarly (see  Appendix \ref{SECTION_II}): for the first term, the derivative is
\begin{equation}\label{multip6}
\begin{split}
\tfrac{d}{dt}\int_{\ssy\Omega}|\phi_{xx}|^2dx=&\,
12\,\alpha\,\mathrm{Re}\left[\int_{\ssy\Omega}
|\phi|^2\,(|\phi_{xx}|^2)_x\;dx\right]
-6\,\alpha\,\mathrm{Re}\left[\int_{\ssy\Omega}
\overline{\phi}_{xx}^2\,\phi_x\,\phi\;dx\right]\\
&-8\,\delta\,\mathrm{Im}\left[\int_{\ssy\Omega}
|\phi_x|^2\,\phi\,\overline{\phi}_{xx}\;dx\right]
-4\,\delta\,\mathrm{Im}\left[\int_{\ssy\Omega}
\phi_x^2\,\overline{\phi}_{xx}\,\overline{\phi}\;dx\right]\\
&-2\,\delta\,\mathrm{Im}\left[
\int_{\ssy\Omega}\phi^2\,\overline{\phi}_{xx}^2\;dx\right],\\
\end{split}
\end{equation}
and for the third term,
\begin{equation}\label{multip7}
\begin{split}
\tfrac{d}{dt}\int_{\ssy\Omega}\phi^2\,\overline{\phi}_x^2\;dx=&\,
3\,\alpha\,\mathrm{Re}\left[\int_{\ssy\Omega}
|\phi|^2\,|\phi_x|^2\,(|\phi|^2)_x\;dx\right]
-6\,\sigma\,\mathrm{Re}\left[\int_{\ssy\Omega}
\overline{\phi}_{xx}^2\,\phi_x\,\phi\;dx\right]\\
&+6\,\alpha\,\mathrm{Re}\left[
\int_{\ssy\Omega}|\phi|^2\,\phi_x
\,\overline{\phi}^2\,\phi_{xx}\;dx\right]
-2\,\rho\,\mathrm{Im}\left[\int_{\ssy\Omega}\phi_{xx}\,\phi
\,\overline{\phi}_x^2\;dx\right]\\
&-2\,\delta\,\mathrm{Im}\left[\int_{\ssy\Omega}|\phi|^2
\,\phi^2\,\overline{\phi}_x^2\;dx\right]
+4\,\rho\,\mathrm{Im}\left[\int_{\ssy\Omega}|\phi_x|^2
\,\overline{\phi}\,\phi_{xx}\;dx\right]\\
&+2\,\rho\,\mathrm{Im}\left[\int_{\ssy\Omega}
\overline{\phi}^2\,\phi_{xx}^2\;dx\right]
+2\,\delta\,\mathrm{Im}\left[\int_{\ssy\Omega}
|\phi|^4\,\overline{\phi}\,\phi_{xx}\;dx\right].
\end{split}
\end{equation}
Constructing the time derivative of the functional
$\mathbf{J_1}(\phi)$ from the left-hand side of
\eqref{multip5}-\eqref{multip7}, we get as a consequence of the
assumption \eqref{cruc2008} on the coefficients,  the equation
\eqref{hirotaenergy}.
\end{proof}
%
%
We shall use Lemma \ref{weaksol1}, to prove
\setcounter{thm}{4}
\begin{thm}\label{globexh2}
Let $\phi\in C([0,T],H^3_{per}(\Omega))\cap C^1([0,T],
L^2(\Omega))$, be the local in time solution of
\eqref{introeq1}-\eqref{introeq2}-\eqref{bc}. Assume that (\ref{cruc2008}) holds, and in addition, that $\rho\neq 0$. Then
$\phi\in C([0,\infty);H^2_{per}(\Omega))
\cap\mathrm{C}([0,\infty);H^1_{per}(\Omega))$,
i.e the solution $\phi$ does not blow up in finite time, in the
$H^m_{per}(\Omega )$ norm, $m=0,1,2$.
\end{thm}
%
%
%
%
%
\begin{proof}
By the conservation law \eqref{law1}, follows the equation
\begin{equation}\label{gs1}
\|\phi\|_{\ssy L^2(\Omega)}^2=\|\phi_0\|_{\ssy
L^2(\Omega)}^2,\quad\forall\,t\in [0,T].
\end{equation}
From the conservation law \eqref{law2},  equation \eqref{gs1}, and inequality
\eqref{gnin}, we have that
\begin{equation}\label{gs2}
\begin{split}
\|\phi_x\|^2_{\ssy L^2(\Omega)}\leq&\,\|\phi_{0x}\|_{\ssy
L^2(\Omega)}^2 +c_{\delta\rho}\,\|\phi_{0}\|_{\ssy L^4(\Omega)}^4
+c_{\delta\rho}\,\|\phi\|_{\ssy L^4(\Omega)}^4\\
\leq&\,M_1+c_{\delta\rho}\,\|\phi_{x}\|_{\ssy L^2(\Omega)}
\,\|\phi_{0}\|_{\ssy L^2(\Omega)},\\
\end{split}
\end{equation}
where $c_{\delta\rho}=|\frac{\delta}{\rho}|$, and $M_1$ depends
only on $\|\phi_0\|_{\ssy H^1(\Omega)}$ and the coefficients. Then,
from \eqref{gs2} and Young's inequality, we derive the estimate
\begin{eqnarray}\label{gs3}
\|\phi_{x}\|_{\ssy L^2(\Omega)}^2\leq\, M_2,
\end{eqnarray}
where $M_2$ is independent of $t$.
\par
The next step is to derive an estimate of $\|\phi_{xx}(t)\|_{\ssy
L^2(\Omega)}$. Note that through the procedure of the proof
of the energy equation \eqref{hirotaenergy}, the integrals
$\mathrm{Re}\left[\int_{\ssy\Omega}|\phi|^2\,(|\phi_{xx}|^2)_x\;dx\right]$
and $\mathrm{Re}\left[\int_{\ssy\Omega}\overline{\phi}_{xx}^2
\,\phi\,\phi_x\;dx\right]$ have been eliminated from the
quantities \eqref{multip5}-\eqref{multip7}. Thus, integration with
respect to time of \eqref{hirotaenergy}, implies the inequality
\begin{equation}\label{h1}
\begin{split}
\|\phi_{xx}(t)\|^2_{\ssy L^2(\Omega)}\leq&
|\rho|^{-1}\left\{\int_{0}^{t}\left(\,|\mathbf{E_1}(s)|
+|\mathbf{E_2}(s)|\,\right)\;ds
+|\mathbf{J_1}(0)|\right\},\\
\mathbf{E_2}(t)=&\,4\,\delta\,\mathrm{Re}\left[
\int_{\ssy\Omega}|\phi|^2|\phi_x|^2dx\right]
+\delta\,\mathrm{Re}\left[\int_{\ssy\Omega}\phi^2
\,\overline{\phi}_x^2\;dx\right].\\
\end{split}
\end{equation}
The various quantities appearing in \eqref{finmultip} can be
estimated with the help of either \eqref{gnin} or \eqref{prop2}.
For example, we may apply \eqref{gnin} for $u=\phi_x$ with
$j=0,\,p=4,\,q=r=2,\,m=1$, and for $u=\phi$ with
$j=1,\,p=3,\,q=r=2,\,m=2$, to obtain respectively, the following
estimates:
\begin{equation}\label{h2}
\|\phi_{x}\|_{\ssy L^4(\Omega)}^2 \leq\,c\,\|\phi_{x}\|_{\ssy
L^2(\Omega)}^{\frac{3}{2}} \,\,\|\phi_{xx}\|^{\frac{1}{2}}_{\ssy
L^2(\Omega)},
\quad
\|\phi\|^{3}_{\ssy L^3(\Omega)}\leq\,c\,\|\phi\|_{\ssy
L^2(\Omega)}^{\frac{5}{4}} \,\,\|\phi_{xx}\|^{\frac{7}{4}}_{\ssy
L^2(\Omega)}.
\end{equation}
The constant $c$ in (\ref{h2}) depends on the constants $M_1,M_2$
of \eqref{gs2} and \eqref{gs3}. Then from Young's inequality and
the estimate \eqref{gs3}, we get that
\begin{equation}\label{h3}
\begin{split}
\int_{\ssy\Omega}|\phi_x|^2\,|\phi|\,|\phi_{xx}|\;dx\leq&
\,\|\phi\|_{\ssy L^{\infty}(\Omega)}\,\|\phi_x\|_{\ssy
L^4(\Omega)}^2
\,\|\phi_{xx}\|_{\ssy L^2(\Omega)}\\
\leq&\,M+c\,\|\phi_{xx}\|^2_{\ssy L^2(\Omega)},\\
\left|\mathrm{Re}\left[\int_{\ssy\Omega}|\phi|^2\,|\phi_x|^2
\,(|\phi|^2)_x\;dx\right]\right|\leq&
\,c\,\int_{\ssy\Omega}|\phi|^3\,|\phi_x|^3\;dx\\
\leq&\,M+c\,\|\phi_{xx}\|^2_{\ssy L^2(\Omega)}.
\end{split}
\end{equation}
The rest of the terms can be estimated in a similar manner,
taking of course into account the estimate of
$\|\phi\|_{\ssy L^{\infty}(\Omega)}$, provided by \eqref{gs3}. The
constant $M$ depends only on the coefficients and
$\|\phi_0\|_{\ssy H^2(\Omega)}$. Eventually, we get the inequality
\begin{equation}\label{h4}
\|\phi_{xx}(t)\|^2_{\ssy L^2(\Omega)}\leq
M+c\int_{0}^t\|\phi_{xx}(s)\|^2_{\ssy L^2(\Omega)}\;ds,
\quad\forall\,t\in[0,T],
\end{equation}
where now, $c$ denotes a constant depending only on the
coefficients and $\|\phi_0\|_{\ssy H^2(\Omega)}$.
\par
The fact that the $H^2_{per}(\Omega)$-solution of
\eqref{introeq1}-\eqref{introeq2}-\eqref{bc} exists for all time
follows essentially by applying the criterion of \cite{Alb}
(see also \cite{Bona} for the modified KdV-equation in
weighted-Sobolev spaces).
\par
Since  $\phi_0\in H^3_{per}(\Omega)$ by assumption, the local existence Theorem~\ref{thmloc} allows us to let  $T_{\max}$ be the maximum
time, such that,  for all $T$ with $0<T<T_{\max}$, the solution
$\phi\in C([0,\infty);H^2_{per}(\Omega))
\cap\mathrm{C}([0,\infty);H^1_{per}(\Omega))$. The claim is that the following alternative holds, \cite[Theorem
2, p. 6]{Alb}:
\begin{equation*}
T_{\max}=\infty,\quad\mbox{or} \sup_{0\leq t<T_{\max}}\left\{\|\phi(t)\|_{\ssy L^{\infty}(\Omega)}+\|\phi_x(t)\|_{\ssy L^{\infty}(\Omega)}\right\}
=\infty.
\end{equation*}
From inequality
\eqref{h4}, we get
\begin{equation}\label{h5}
\|\phi_{xx}(\cdot,t)\|_{\ssy L^2(\Omega)}^2\leq
\,M\,e^{ct},\quad\forall\,t\in [0,T_{\max}).
\end{equation}
Then, it follows that $T_{\max}=\infty$. This can be seen by
adapting in our case the last lines of \cite[Theorem 2, pg.
8]{Alb}: For instance, if $T_{\max}<\infty$, then there exists a
finite constant $M_3$ such that
\begin{equation*}
\|\phi(\cdot,t)\|_{\ssy H^2_{per}(\Omega)} \leq\,M_3,\quad\forall\,t\in
(0,T_{\max}).
\end{equation*}
We may  apply next,  the local existence result given by
Theorem~\ref{thmloc} with the initial condition
$$\widetilde{\phi_0}(\cdot):=\phi (\cdot, T_{\max}-\frac{1}{2}T(M_3)).$$
With this initial condition, Theorem~\ref{thmloc} implies the existence of a solution in
$C([0,{\widetilde{T}}];H^2_{per}(\Omega))$ with $\widetilde{T}$
being at least $T_{\max}+\frac{1}{2}\,T(M_3)$. This $\widetilde{T}$  contradicts obviously
the maximality of $T_{\max}$.  The contrary  of the alternative criterion follows again by the Sobolev
embeddings \eqref{prop1}.
\end{proof}
%
%
%
%
%
%
\setcounter{remm}{5}
\begin{remm}\label{issue2} (Conservation laws and energy equations).
{\em The derivation of the conservation law \eqref{law2}, and of the  $H^2$-energy equation depends
heavily on the assumption \eqref{cruc2008} on the  coefficients. Sharing the same difficulty with KdV-type equations, 
it is not
easy to recover explicit formulas for the possible invariants due
to the rapid proliferation of terms with increasing rank
\cite{RTem88}. As it was already mentioned in the proof of
Theorem~\ref{globexh2}, the assumption \eqref{cruc2008} has a strong
effect in the global solvability of
\eqref{introeq1}--\eqref{bc}-\eqref{introeq2}, since it allows for the
elimination of such terms. To the best of our knowledge, it is unclear if generalized ENLS equations of the form \eqref{geneq1} conserve a similar energy functional to (\ref{law2}), as the NLS counterparts of the form 
$u_t-\mathrm{i}\rho u_{xx}+f(|u|^2)u=0$, with a general class of nonlinearities  $f$ --see \cite{cazS,cazh}. 
As an example for \eqref{geneq1}, one
may consider the case of a generalized Hirota equation with power
nonlinearities $\beta(s)=s^p=\gamma(s)$, or
$\beta(s)=s^p,\,\gamma(s)=s^q$, where $p,q$ are non-negative
integers. This case obviously satisfies the conservation law
\eqref{law1}. However, it does not seem to have similar energy
invariants with those of NLS and generalized KDV equations, with such a type of nonlinearities, \cite[p.12]{Bona}.

\par
Finally, we stress that the derivation of the conservation laws (\ref{law1})-(\ref{law2}), and of the $H^2$-energy equation (\ref{hirotaenergy}), can be produced even under weaker assumptions on the initial data.  Actually, a similar approximation procedure to that followed in the proof of Lemma \ref{weaksol1}, can be used to prove the conservation laws (\ref{law1})-(\ref{law2}), even when the initial condition $\phi_0\in H^1_{per}(\Omega)$, and the associated solution  $\phi\in C([0,T],H^1_{per}(\Omega))\cap C^1([0,T], H^{-2}_{per}(\Omega))$. Furthermore, revisiting the formulas \eqref{devfunct}-\eqref{intt7} under the assumption that the initial condition $\phi_0\in H^1_{per}(\Omega)$, it can be shown that the $H^2$-energy equation is valid for weak solutions $\phi\in C([0,T],H^2_{per}(\Omega))\cap C^1([0,T], H^{-1}_{per}(\Omega))$. The choice of initial condition $\phi_0\in H^3_{per}(\Omega)$ in the proofs has been made for the simplicity of the presentation; under this assumption, it is only necessary to present the approximation scheme only for the proof of Lemma \ref{weaksol1}.} 
\end{remm}
\section{Exact 
solitary wave solutions}
\label{SECTION_IV}
This section is devoted to the investigation of appropriate conditions for the existence of 
solitary wave solutions for (\ref{introeq1}). The analysis
employs a travelling-wave ansatz, and the relevant phase-plane considerations 
of \cite{DJF95b}-see also \cite{DJF95a}. 
In particular, we seek travelling wave solutions solutions in the form:
\begin{eqnarray}
\label{eq3.1}
\phi(x,t)=\mathrm{e}^{\mathrm{i} (k_0 x-\omega_0 t)}F(\xi), \qquad \xi=x-\upsilon t.
\end{eqnarray}
Here, $F(\xi)$ is 
an unknown envelope function (to be determined) which
will be assumed to be 
real. The real parameters $k_0$ and $\omega_0$ are the wavenumber and frequency of the travelling wave, respectively, 
while the parameter $\upsilon$ stands for the group velocity of the wave in the $x$-$t$ frame and is connected with the shift of the original group velocity $v_g$ in the context of optical fibers, see \cite{DJFa1}. The parameters $k_0$, $\omega_0$ and $\upsilon$ are 
unknown parameters which shall be determined. 

Substituting 
(\ref{eq3.1}) into (\ref{introeq1}), we obtain 
the equation:
\begin{eqnarray}
\label{eq3.4a}
&-\upsilon F'-\mathrm{i}\omega_0 F+3\alpha|F|^2(F'+\mathrm{i}k_0 F)-\mathrm{i}\rho (F''+2\mathrm{i}k_0 F'-k_0^2F)\nonumber\\
&+\sigma (F'''+3\mathrm{i}k_0 F''-3k_0^2F'-\mathrm{i}k_0^3F)-\mathrm{i}\delta|F|^2F=0,
\end{eqnarray}
where primes denote derivatives with respect to $\xi$. 
Separating the real and imaginary parts of (\ref{eq3.4a}), we derive the system of equations
\begin{eqnarray}
\label{eq3.4}
&F'''-\frac{\upsilon+3k_0^2\sigma-2\rho k_0}{\sigma}F'+\frac{3\alpha}{\sigma}F^2F'=0,\\
\label{eq3.5}
&F''-\frac{\omega_0-\rho k_0^2+\sigma k_0^3}{3k_0\sigma-\rho}F+\frac{3\alpha k_0-\delta}{3 k_0\sigma-\rho}F^3=0,
\end{eqnarray}  
under the assumptions
\begin{eqnarray}
\label{eq3.6}
\sigma\neq 0,\;\;\mbox{and}\;\;\rho\neq 3k_0\sigma.
\end{eqnarray}
To derive a consistent (i.e., not overdetermined) system, 
Eq.~(\ref{eq3.5}) should be differentiated once again. Hence, the system of Eqs.~(\ref{eq3.4})-(\ref{eq3.5})
will be consistent under the conditions
\begin{eqnarray}
\label{eq3.9}
\frac{\upsilon+3k_0^2\sigma-2\rho k_0}{\sigma}&=&\frac{\omega_0-\rho k_0^2+\sigma k_0^3}{3k_0\sigma-\rho}=\lambda,\;\;\mbox{and}\\
\label{eq3.10}
\frac{\alpha}{\sigma}&=&\frac{3\alpha k_0-\delta}{3k_0\sigma-\rho}=\mu,
\end{eqnarray}
where $\lambda$ and $\mu$ are real nonzero constants. 

It is interesting to recover the balance condition (\ref{cruc2008}) anew, this time as a necessary condition for the existence of travelling-wave solutions of the form (\ref{eq3.1}).  In particular, the consistency condition (\ref{eq3.10}) implies directly that $\alpha\rho=\sigma\delta$ for any $k_0\in\mathbb{R}$.  Furhtermore, by solving (\ref{eq3.9}) in terms of $\omega_0$, we find that the equation (\ref{introeq1}) supports travelling-wave solutions of the form (\ref{eq3.1}) for any $k_0\neq\frac{\rho}{3\sigma}$,  if the condition (\ref{cruc2008}) holds, and $\omega_0$, $\upsilon$ satisfy
\begin{eqnarray}
\label{eq3.4aa}
\omega_0=\frac{1}{\sigma}\left[(\upsilon+3k_0^2\sigma-2k_0\rho)(3k_0 \sigma-\rho)+\sigma(\rho k_0^2-\sigma k_0^3)\right].
\end{eqnarray}
In the particular case of travelling wave solutions  with wavenumber $k_0=0$, we see directly from (\ref{eq3.4aa}) that $\omega_0$, $\upsilon$ should satisfy the relation
\begin{eqnarray}
\label{eq3.4A0}
\upsilon\rho=-\omega_0\sigma.
\end{eqnarray}

With the consistency conditions (\ref{eq3.9}) and (\ref{eq3.10}) in hand, 
the system 
(\ref{eq3.4})-(\ref{eq3.5}) is equivalent to the unforced and undamped Duffing oscillator equation 
\begin{eqnarray}
\label{eq3.14}
F''+\lambda F+\mu F^3=0,
\end{eqnarray}
with the associated Hamiltonian
\begin{eqnarray}
\label{eq3.15}
\mathcal{H}(F,F')=\frac{F'^2}{2}+\lambda\frac{F^2}{2}+\mu\frac{F^4}{4}.
\end{eqnarray}
The exact 
solitary wave solutions 
of (\ref{introeq1}), as well as their dynamics will then be determined by the corresponding solutions and dynamics of (\ref{eq3.14}) [or (\ref{eq3.15})], according to the analysis carried out e.g. in \cite{DJF95b}. 
To reduce the number of parameters involved in this analysis we use 
the change of variables $t\rightarrow\rho t$, and the transformation 
$(\delta/\rho)|\phi|^2\rightarrow s|\phi|^2$,
and 
reduce equation (\ref{introeq1})
to the following form
\begin{equation}\label{eq3.18}
\phi_{t}+3A s|\phi|^2\phi_x-{\rm
i}\phi_{xx}+B\phi_{xxx} -{\rm
i}s |\phi|^2\,\phi=0,
\end{equation}
where the constants $A$, $B$ and $s$ are defined as
\begin{eqnarray}
\label{eq3.19}
A=\frac{\alpha}{\delta},\;\;B=\frac{\sigma}{\rho},\;\;s=\mathrm{sign}\left(\frac{\delta}{\rho}\right).
\end{eqnarray}
%
Accordingly, Eqs.~(\ref{eq3.9}) and (\ref{eq3.10}) respectively become, 
\begin{eqnarray}
\label{eq3.21}
\lambda&=&\frac{\upsilon+3k_0^2B-2k_0}{B}=\frac{\omega_0-k_0^2+Bk_0^3}{3k_0 B-1},\\
\label{eq3.22}
\mu&=&\frac{A s}{B}=\frac{3A sk_0 -s}{3Bk_0-1},
\end{eqnarray}
%
and the conditions for the existence of 
travelling-wave solutions can be restated for 
(\ref{eq3.18}) as: 
\begin{eqnarray}
\label{eq3.23}
A=B\;\;\mbox{and}\;\;\;\omega_0=k_0^2-Bk_0^3+\frac{(\upsilon+3k_0^2B-2k_0)(3k_0 B-1)}{B}.
\end{eqnarray}
The first condition $A=B$ in (\ref{eq3.23}) is a consequence of the requirement 
$\alpha\rho=\sigma\delta$, i.e., the condition (\ref{cruc2008}). The second condition in (\ref{eq3.23}) is obtained 
by solving (\ref{eq3.21}) in terms of $\omega_0$.

We now proceed with the travelling wave solutions, which can be clasified according to 
Eqs.~(\ref{eq3.14}) and (\ref{eq3.15})
as follows:\newline
{\em Bright Solitons.} Bright solitons 
correspond to the symmetric homoclinic orbits, on the saddle point $(F,F')=(0,0)$,
which
are described by:
\begin{eqnarray}
\label{eq3.24}
F(\xi)=\pm\sqrt{2|\lambda|}\,\mathrm{sech}\left(\sqrt{|\lambda|}\xi\right),
\end{eqnarray}
where the coefficients $\lambda$ and $\mu$ given in (\ref{eq3.21}) and (\ref{eq3.22}) satisfy
\begin{eqnarray}
\label{eq3.25}
\mu=1,\;\;\mbox{and}\;\;\lambda<0.
\end{eqnarray}
\newline
{\em Dark Solitons.} Dark solitons 
correspond to the symmetric heteroclinic orbits connecting the saddle points $(F,F')=(\pm \sqrt{|\lambda/\mu|},0)$ which are described by:
\begin{eqnarray}
\label{eq3.26}
F(\xi)=\pm\sqrt{|\lambda|}\,\mathrm{tanh}\left(\sqrt{\frac{|\lambda|}{2}}\xi\right),
\end{eqnarray}
where the coefficients $\lambda$ and $\mu$ satisfy
\begin{eqnarray}
\label{eq3.27}
\mu=-1,\;\;\mbox{and}\;\;\lambda<0.
\end{eqnarray}

Another interesting observation is that the ENLS (\ref{introeq1}) does not support standing-wave solutions
of the form $\phi(x,t)=F(x)\exp(-i\omega_0 t)$.
Such a solution 
is a special case of the solutions (\ref{eq3.1}) with $k_0=\upsilon=0$. However, it follows from (\ref{eq3.4A0}) that if $\upsilon=0$, then $\omega_0=0$. Hence, in this case $\phi$ cannot be time-periodic, but is a time-independent, steady-state solution.  
\setcounter{remm}{0}
\begin{remm}\label{issue4} 
(Conditions for the existence of exact soliton solutions). {\em Recovering (\ref{cruc2008}) as a necessary condition for the existence of exact soliton solutions, it is also interesting to find that the additional conditions  (\ref{eq3.6}) and (\ref{eq3.4A0}) elucidate further the role of the linear dispersion strengths $\rho$ and $\sigma$ 
on the support of such waveforms by the ENLS (\ref{introeq1}). For instance, when $k_0\neq 0$, the ENLS equation supports soliton solutions when both second and third-order dispersion are present, i.e., $\rho\neq 0$ and $\sigma\neq 0$, excluding  solutions of wavenumber $k_0\neq\frac{\rho}{3\sigma}$. Then, the dispersion effects should be counterbalanced 
by the nonlinear effects as described by the balance condition (\ref{cruc2008}). On the other hand, 
for traveling-wave solutions of wavenumber $k_0=0$, in addition to the balance condition (\ref{cruc2008}), the
constraint (\ref{eq3.4A0}) should be satisfied. 
Since (\ref{eq3.6}) assumes that $\sigma\neq 0$, the case $\rho=0$ is allowed from (\ref{eq3.4A0}), only when $\omega_0=0$ and $\upsilon\neq 0$, in order to get a traveling wave solution. In this case, 
the system supports traveling wave solutions of the form $F(x-\upsilon t)$. However, in this case we may recover again the cmKDV limit. Indeed, for $\sigma\neq 0$ and $\rho=0$, the balance condition (\ref{cruc2008}) implies that $\delta=0$.  
Therefore, if $k_0=0$ traveling waves may exist in both cases $\rho \ne 0$ and $\rho =0$, with the 
latter case corresponding to the so-called ``zero-dispersion point'' (ZDP) in the context of nonlinear fiber optics \cite{CPag95}. Finally, the case $\rho\neq 0$, $\sigma=0$, leads via the balance condition~(\ref{cruc2008}) 
to $\alpha=0$, i.e., to the NLS-limit.
}
\end{remm}	
\subsection{Numerical Study 1: Solitary wave solutions} In this section, we initiate our numerical studies by testing the conditions for the existence of bright and dark solitons. 
We focus on the case of a stationary carrier [with wave number $k_0=0$ in Eq.~(\ref{eq3.1})]; then, it follows from (\ref{eq3.4A0}) and (\ref{eq3.19}), that
travelling wave solutions (\ref{eq3.1}) exist if
\begin{eqnarray}
\label{eq3.23a}
A=B\;\;\mbox{and}\;\;\;\frac{\upsilon}{B}=-\omega_0.
\end{eqnarray}
For bright soliton solutions  
the parameters $\mu$,
$\upsilon,\,B,\,$ and $\omega_0$ are such that
\begin{eqnarray}
\label{eq3.25a}
\mu=1,\;\;\omega_0<0\;\;\mbox{and}\;\;\frac{\upsilon}{B}>0,
\end{eqnarray}
while for dark solitons
\begin{eqnarray}
\label{eq3.27a}
\mu=-1,\;\;\omega_0>0\;\;\mbox{and}\;\;\frac{\upsilon}{B}<0.
\end{eqnarray}
\begin{figure}
\label{fig2}
\begin{center}
    \begin{tabular}{cc}
   \includegraphics[scale=0.5]{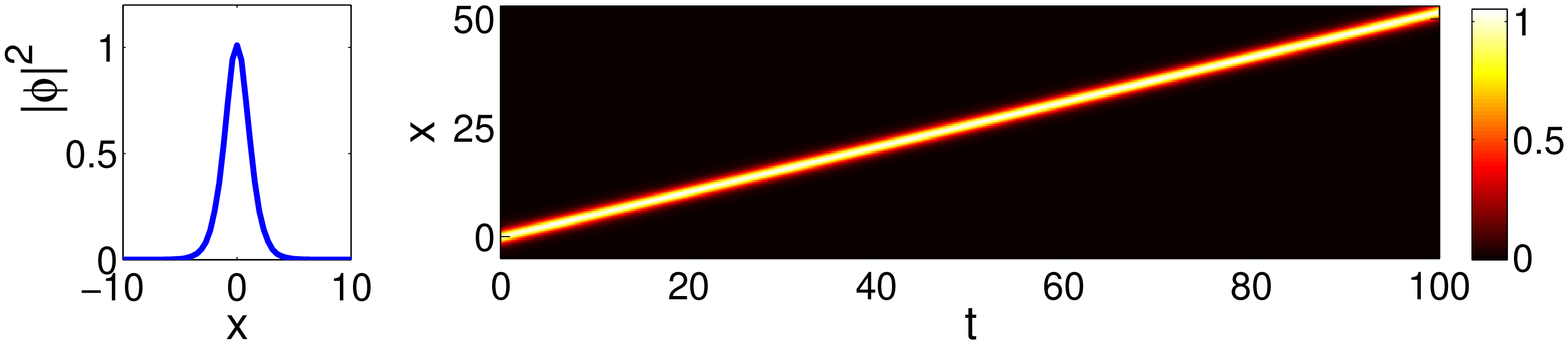}\\
    \includegraphics[scale=0.5]{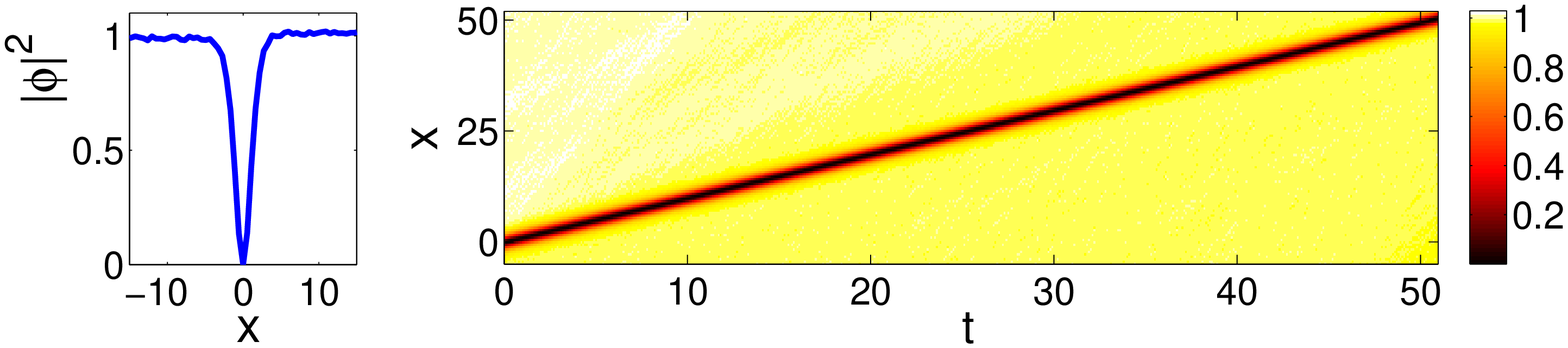} 
    \end{tabular}
\caption{\small{
Top panel: The initial density $|\phi|(x,0)|^2$ (left panel)
of a bright soliton [cf. 
Eq.~(\ref{inBS1})], and a contour plot showing its evolution 
under a random small-amplitude 
perturbation. 
Parameters: $A=B=1$, $\omega_0=-0.5$, $\upsilon=0.5$, $s=1$. 
Bottom panel: Same as the top panel, but for a dark soliton 
[cf. Eq.~(\ref{inDS1})] 
with parameters: $A=B=-1$, $\omega_0=1$, $\upsilon=1$, $s=-1$.}}
\end{center}
\end{figure}

The 
top left panel of Fig.~1 
presents the 
density profile $|\phi|^2$ of the initial condition 
\begin{eqnarray}
\label{inBS1}
\phi_0(x)=\mathrm{sech}(0.5x)+\epsilon\cdot\mathrm{noise},\;\;\epsilon=10^{-3},
\end{eqnarray}
which describes a bright soliton perturbed by a small-amplitude noise. The 
top right panel of Fig.~1 is a contour plot showing the evolution of the density of the above state. 
The unperturbed initial bright soliton $u_{bs}(x)=\mathrm{sech}(0.5x)$ 
has the form of (\ref{eq3.24}), for the following choice of parameters suggested from (\ref{eq3.23a})-(\ref{eq3.27a}):  $A=B=1$, and $\mu=1$. For this choice of $\mu$, the definition 
in (\ref{eq3.22}) implies that $s=1$. We have also chosen $\omega_0=-0.5$, and since  $k_0=0$, it turns out from (\ref{eq3.21}) that $\lambda=0.5$. 

The result shown in the top panel of Fig.~1 
confirms the existence of the solitary pulse which should be traveling with speed $\upsilon=0.5$ (as it is expected from (\ref{eq3.23a})); this is confirmed by the simulation. 
Additionally, we find 
that the solitary pulse is robust under this small-amplitude random perturbation; in fact, we have found that the 
solitary wave persists even for noise amplitudes up to $\epsilon=0.1$.
\begin{figure}
\label{fig3}
\begin{center}
    \begin{tabular}{cc}
   \includegraphics[scale=0.5]{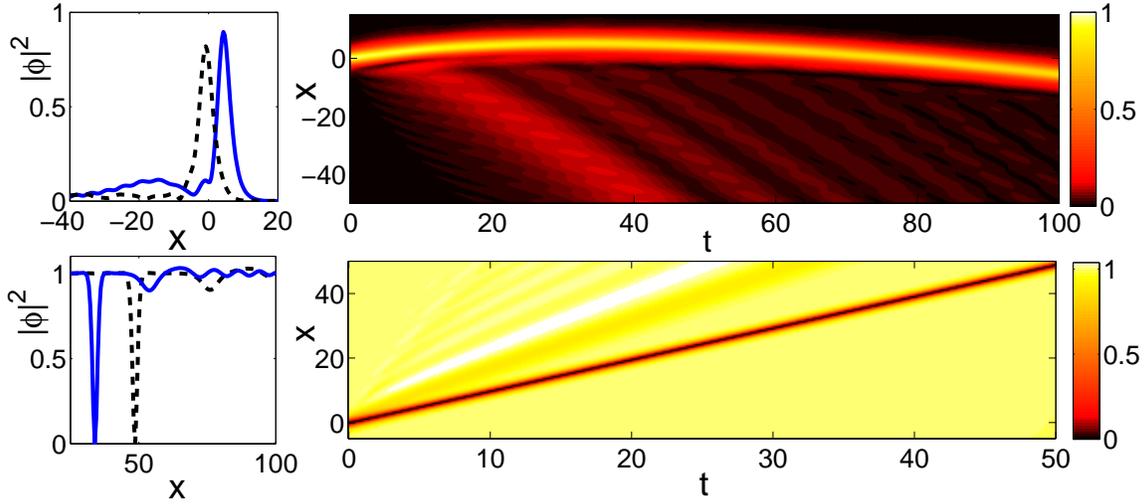}\\
    \end{tabular}
\caption{\small{
Top panel: Contour plot (right panel) showing the evolution of a
bright solitary wave corresponding to the initial condition given in Eq.~(\ref{inBS1}) with $\epsilon=0$, 
for $A/B=0.76$.
The bright pulse is accompanied by a small-amplitude radiation as depicted by  both the solid (blue)- and the dashed (black) line
 in the left panel, showing the soliton profile at $t=20$ and $t=60$ respectively.
Bottom panel: same as 
top panel, but for a dark solitary wave, with the initial condition 
given by Eq.~(\ref{inDS1}) with $\epsilon=0$,
for $A=-1$, $B=-0.8$, $s=-1$. The 
left panel shows the 
density profile at time $t=35$ [dashed (black) curve] and at 
$t=50$ 
[solid (blue) curve]. 
The radiative part emerging on top of the background moves faster than the solitary wave so that, eventually, 
the solitary wave is separated from the radiation. 
}
}
\end{center}
\end{figure}
We have also studied the case of a dark soliton; 
corresponding numerical results are presented in the bottom panels of Fig.~1. 
As before, the initial condition is taken to be a 
dark soliton perturbed by a small-amplitude 
uniformly-distributed noise, namely:
\begin{eqnarray}
\label{inDS1}
\phi_0(x)=\mathrm{tanh}(0.5x)+\epsilon\cdot\mathrm{noise},\;\;\epsilon=10^{-3}.
\end{eqnarray}
The density of the initial condition (\ref{inDS1}) is shown in the bottom left  panel of Fig.~1. 
The unperturbed initial dark soliton $u_{ds}(x)=\mathrm{tanh}(0.5x)$ is taken from 
Eq.~(\ref{eq3.26}), for the following choice of parameters suggested from (\ref{eq3.23a})-(\ref{eq3.27a}): 
$A=B=-1$ and $\mu=-1$. For this choice of $\mu$, (\ref{eq3.22}) implies 
that $s=-1$. We  have also chosen $\omega_0=1$, and since $k_0=0$, it turns out from (\ref{eq3.21}) that $\lambda=1$. 
The time evolution of the density corresponding to the initial data (\ref{inDS1}) is shown in the contour-plot of the 
bottom right panel of Fig.~1. 
As in the previous (bright soliton) case, the existence of the dark solitary pulse 
travelling with speed $\upsilon=1$ is confirmed [as 
predicted 
by (\ref{eq3.23a})]; furthermore, the solitary wave is found to be robust 
under the random perturbation.

Next, we consider the case where the condition for the existence of traveling pulses is violated, 
i.e., $A\neq B$. Examples corresponding to such a case are shown 
in Fig.~2, for bright (top panels) and dark (bottom panels) solitons.
In particular, the top right panel of Fig.~2 presents the evolution of the density of 
a bright soliton with the initial form $u_{bs}(x)=\mathrm{sech}(0.5x)$ [corresponding to 
(\ref{inBS1}) with $\epsilon=0$], with 
parameter values $A=1$ and $B=1.3$ (i.e. $A/B=0.76$).
It is observed that the bright soliton is a traveling one, and emits radiation during its evolution;
the radiation is stronger at early times and becomes smaller for later times, tending to disperse away from the pulse; this is depicted in the top left panel,
where the soliton profile is shown for  $t=20$ (solid line) and $t=60$ (dashed line).
Such a behavior has also been observed in the numerical results of \cite{GromTal2000}, and analyzed in \cite{YangPeli2, Yang1}. 
As seen in the contour plot, the soliton initially starts moving towards the positive $x$-direction, but eventually travels with an almost
constant velocity towards the negative $x$-direction.


The emission of radiation by the soliton, 
is also observed in the case of a dark solitary wave, as shown in the bottom panels of Fig.~2; 
nevertheless, notable differences can be identified in comparison with the bright soliton dynamics, as is explained below.
. 
The left panel of Fig.~2 shows snapshots of the density at two different time instants ($t=35$ and $t=50$), while 
the initial form of the solitary wave is 
$u_{ds}(x)=\mathrm{tanh}(0.5x)$ 
(with $A=-1$, $B=-0.8$, $s=-1$); on the other hand, the right panel shows the evolution of the density. 
A significant difference with the bright solitary wave is the absence of the turning 
effect observed in the propagation of the bright soliton: the dark pulse is moving continuously  
to the positive $x$-direction. As a relevant aside, we also note that ahead
of the dark soliton we observe the traces of an apparent dispersive shock
wave~\cite{DSW}, while in the case of the bright soliton, the relevant
radiative wavepackets can only be observed behind the solitary
wave. It is also observed that, contrary to the bright solitary wave case, 
the initial condition tends to a single dark solitary wave moving with an almost constant velocity and a separate 
radiative part. The latter, travels on top of the continuous-wave background and moves with a velocity larger
than that of the dark pulse (the radiative small-amplitude waves in this case travel with the speed of sound), 
so they eventually separate from each other. Another notable feature is that radiation is emitted in the 
rear (front) of the bright (dark) solitary wave, due to the different sign of the third-order dispersion parameter $B$: 
it is positive (negative) for the bright (dark) pulse and, thus, the emission of radiation is directional. 

%

\section{Modulation instability}

\label{SECTION_V}
In this section, we will investigate conditions for the modulation instability (MI) 
of plane-wave solutions of Eq.~(\ref{eq3.18}) of the form 
\begin{eqnarray}
\label{eq3.61}
\phi^{(0)}(x,t)=\phi_0\mathrm{e}^{\mathrm{i}(kx-\omega t)}. 
\end{eqnarray}
Here, the plane wave is assumed to have 
a constant amplitude $\phi_0\in\mathbb{R}$, a wavenumber $k$ and frequency $\omega$. 
We will show that for $k\neq 0$ 
there exist 
essential differences on the 
MI conditions between the ENLS equation (\ref{eq3.18}) and its NLS limit $A=B=0$.

We start by substituting 
the solution (\ref{eq3.61}) into 
(\ref{eq3.18}). We thus obtain 
the following dispersion relation $\omega(k;\phi_0^2)$: 
\begin{eqnarray}
\label{eq3.62}
\omega=-s\phi_0^2+3Ask\phi_0^2+k^2-Bk^3.
\end{eqnarray}  
We now consider a perturbation of the solution (\ref{eq3.61}) of the form:
\begin{eqnarray}
\label{eq3.63}
\phi(x,t)=\phi^{(0)}(x,t)+\phi^{(1)}(x,t),\;\;\phi^{(1)}(x,t)\in\mathbb{C},\;\;\mbox{for all $x\in\mathbb{R},\;t\in\mathbb{R}$},
\end{eqnarray}
where $\phi^{(1)}(x,t)$ is assumed to be small in the sense
\begin{eqnarray}
\label{eq3.64}
|\phi^{(1)}(x,t)|< < |\phi^{(0)}|,\;\;\;\;\mbox{for all $x\in\mathbb{R},\;t\in\mathbb{R}$}, 
\end{eqnarray}
%
while $\phi^{(1)}(x,t)$ is taken to be of the form:
%
\begin{eqnarray}
\label{eq3.65}
\phi^{(1)}(x,t)=\mathrm{e}^{\mathrm{i}(kx-\omega t)}\phi_1(x,t).
\end{eqnarray}
%
Then, $\phi_1$ can be decomposed 
in its real and imaginary parts as:
\begin{eqnarray}
\label{eq3.66}
\phi_1(x,t)=U(x,t)+\mathrm{i}W(x,t), 
\end{eqnarray}
and $U$ and $W$ are considered to be harmonic, i.e.,  
  \begin{equation}
\label{eq3.67}
\left(
  \begin{array}{cc}
    U \\
    W\\
  \end{array}
\right)
=
\left(
  \begin{array}{c}
    U_0 \\
    W_0 \\
  \end{array}
\right)\mathrm{e}^{\mathrm{i}(Kx-\Omega t)}+c.\;c.\;,\;\;U_0,\;W_0\in\mathbb{R}, 
\end{equation}
where $K$ and $\Omega$ are the wave number and frequency of the perturbation, respectively.

\begin{thm}
\label{MITHk}
We consider the equation (\ref{eq3.18}) with $A,B\in\mathbb{R}$. The plane-wave solutions (\ref{eq3.61})
are modulationally unstable (i.e., the perturbations $\phi_1(x,t)$ defined by (\ref{eq3.65}), 
(\ref{eq3.66}), and (\ref{eq3.67}) 
are growing at exponential rate), if and only if
\begin{eqnarray}
\label{SW04k}
1\neq 3Bk,\;\;\mbox{and}\;\;K^2<2s\phi_0^2\left(\frac{1-3kA}{1-3kB}\right).
\end{eqnarray} 
\end{thm}
\begin{proof} 
Substituting the perturbed solution (\ref{eq3.63}) in equation (\ref{eq3.18}), linearising 
as per the smallness condition (\ref{eq3.64}), and 
using the dispersion relation (\ref{eq3.62}), 
we obtain the following equation for $\phi_1(x,t)$:
%
\begin{eqnarray}
\label{eq3.68}
\phi_{1t}&+&3As\phi_0^2[\phi_{1x}+\mathrm{ik}(\phi_1+\bar{\phi}_1)]-\mathrm{i}(1-3kB)\phi_{1xx}
\nonumber\\
&+&(2k-3Bk^2)\phi_{1x}+B\phi_{1xxx}-\mathrm{i}s\phi_0^2(\phi_1+\bar{\phi}_1)=0.
\end{eqnarray}
Inserting (\ref{eq3.66}) into (\ref{eq3.68}), we derive the following equations for $U$ and $W$, 
\begin{eqnarray}
\label{eq3.69}
&&U_t+3As\phi_0^2U_x+(1-3kB)W_{xx}+(2k-3Bk^2)U_x+BU_{xxx}=0,\\
\label{eq3.70}
&&W_t+3As\phi_0^2(W_x+2kU)-(1-3kB)U_{xx}+(2k-3Bk^2)W_x+BW_{xxx}-2s\phi_0^2U=0.
\end{eqnarray}
Next, 
substitution of the expression (\ref{eq3.67}) for $U$ and $W$ into (\ref{eq3.69})-(\ref{eq3.70}) 
yields the following algebraic system for $U_0$ and $W_0$,
\begin{eqnarray}
\label{eq3.71}
&&\left\{-\mathrm{i}\Omega+
\left[3As\phi_0^2+k(2-3Bk)\right]\mathrm{i}K-\mathrm{i}BK^3\right\}U_0-(1-3KB)K^2W_0=0,\\
\label{eq3.72}
&&\left[2s\phi_0^2(3Ak-1)+(1-3kB)K^2\right]U_0
+ \left\{-\mathrm{i}\Omega+
\left[3As\phi_0^2+k(2-3Bk)\right]\mathrm{i}K-\mathrm{i}BK^3\right\}W_0=0.
\end{eqnarray}
Seeking 
for nontrivial solutions $U_0$ and $W_0$ of the system (\ref{eq3.71})-(\ref{eq3.72}), we require the 
relevant determinant to be zero; this way, we 
obtain the following dispersion relation:
%
\begin{eqnarray}
\label{eq3.73}
&&\left\{\Omega-
\left[3As\phi_0^2+k(2-3Bk)\right]K+BK^3\right\}^2\nonumber\\
&&=K^2\left[(1-3kB)^2K^2-2s\phi_0^2(1-3Ak)(1-3kB)\right].
\end{eqnarray}
From (\ref{eq3.73}) the MI condition for plane waves (\ref{eq3.61}) readily follows. For instance, assuming that $1\neq 3kB$, it is required that 
\begin{eqnarray}
\label{eq3.74}
(1-3kB)^2K^2<2s\phi_0^2(1-3Ak)(1-3kB),
\end{eqnarray}
or equivalently, (\ref{SW04k}).
\end{proof}

It is relevant to consider certain limiting cases of (\ref{eq3.73}). First, if $k=0$ (i.e., the plane wave is stationary), 
Eq.~(\ref{eq3.73}) is reduced to the form:
\begin{eqnarray}
\label{eq3.55}
(\Omega-3AKs\phi_0^2+BK^3)^2=K^2(K^2-2s\phi_0^2).
\end{eqnarray}
The above equation has two solutions
\begin{eqnarray}
\label{eq3.60}
\Omega=3AKs\phi_0^2-BK^3\pm K\sqrt{K^2-2s\phi_0^2}.
\end{eqnarray}
The 
frequency $\Omega$
becomes complex (i.e., $\Omega=\Omega_{\rm r}+i\Omega_{\rm i}$) if and only if 
\begin{eqnarray}
s=1, \qquad K^2-2\phi_0^2<0.
\label{nlsmi}
\end{eqnarray}
The above are the same as the conditions for MI of plane wave solutions of the 
NLS equation ($A=B=0$); note that in the case $s=-1$, (\ref{eq3.60}) implies that 
the plane wave is always modulationally stable, as in the (defocusing) NLS case. Thus, for $k=0$, the only difference 
from the NLS limit is that in the ENLS case 
there exists a nonvanishing real (oscillatory) part of the perturbation frequency, i.e., 
$\Omega_{\mathrm{r}}=3AKs\phi_0^2-BK^3$, 
as suggested by (\ref{eq3.60}). Additionally, it is observed that in the NLS limit  
$A=B=0$, 
the well-known result (see, e.g., \cite{CPag95,KodHas87}) 
that the MI condition for the focusing NLS model ($s=1$) does not depend on $k$ is recovered. 

On the other hand, in the framework of the ENLS with 
$A,B\neq 0$, and $k\neq 0$, condition (\ref{SW04k}) reveals 
some important differences 
between 
the ENLS equation with its NLS limit. First, the instability criterion 
 is modified for the focusing case ($s=+1$), 
as seen by a direct comparison of (\ref{SW04k}) and (\ref{nlsmi}). Secondly, 
and 
even more importantly, condition (\ref{SW04k}) shows that plane waves (\ref{eq3.61}) 
can be modulationally unstable {\it also} in the defocusing ($s=-1$) case, contrary to what happens 
in the defocusing NLS limit where such an instability is absent. 
Thus concluding, for the ENLS model, 
continuous waves with $k=0$
can be modulationally unstable 
only in the focusing case $s=1$, 
while plane waves 
of the form of equation (\ref{eq3.61})
with $k\neq 0$, can be modulationally unstable in both the focusing $s=1$ and the defocusing case $s=-1$ 
(for an appropriate choice of parameters).   

Let us now define the instability band, as 
the closed interval $B_I=[K_{\mathrm{min}},K_{\mathrm{max}}]$ (for fixed $\phi_0$) 
for the wavenumbers of the perturbations of (\ref{eq3.65}) with the following property: 
For every $K\in B_I$, the condition (\ref{SW04k}) is satisfied, and the solutions of the quadratic equation
(\ref{eq3.73}) in terms of $\Omega$ are complex conjugates. The imaginary parts of these solutions 
are responsible for the exponential growth of the perturbations (\ref{eq3.65}),
implying  the  MI of the plane wave (\ref{eq3.61}). In particular, 
$\mathrm{Im}(\Omega)=\Omega_{i}$ can be
plotted as a function of $K\in B_I$. The possible maximum of this curve, 
denoted by $\Omega_{\mathrm{crit}}$, attained on $K_{\mathrm{crit}}\in B_I$ is associated with 
the maximum growth of the unstable perturbations. 

We conclude the discussion on the MI effects of the solutions (\ref{eq3.61}) with $k\neq
0$, by 
explicitly showing, that a choice of a wave number $K=K_0\in B_I$ for the perturbations (\ref{eq3.65}),
implies the excitation of the integer multiples of $K_0$ for the unstable wave solutions, i.e., the wave numbers of the 
unstable plane wave solutions are $k\pm nK_0$, $n\in \mathbb{N}$.
%
To show this, 
we start by rewriting 
the perturbed plane wave (\ref{eq3.60}) accounting the complex conjugates (c.c.) of the perturbations (\ref{eq3.65}), i.e., 
\begin{eqnarray}
\label{eq3.75}
\phi(x,t)=\left[\phi_0+U_0\mathrm{e}^{\mathrm{i}(Kx-\Omega t)}+c.c. +\mathrm{i}(W_0\mathrm{e}^{\mathrm{i}(Kx-\Omega t)}+c.c)\right]\mathrm{e}^{\mathrm{i}\theta},\;\;\theta=kx-\omega t.
\end{eqnarray}
We consider the unstable wave number $K_0\in  B_I$ and the associated imaginary part $\Omega_i=\Omega_i(K_0)$. Then,
(\ref{eq3.75}) may be written as $\phi(x,t)=R\mathrm{e}^{\mathrm{i}\varphi}\mathrm{e}^{\mathrm{i}\theta}$, 
where $R$ and $\varphi$ are 
respectively given by:
\begin{eqnarray}
\label{eq3.77}
R&=& \left[\phi_0^2+4\phi_0U_0\mathrm{e}^{\Omega_it}+4(U_0^2+W_0^2)\mathrm{e}^{2\Omega_it}\cos^2(K_0x)\right]^{\frac{1}{2}},
\end{eqnarray}
\begin{eqnarray}
\label{eq3.79}
\varphi=\arctan\left[\frac{2W_0\mathrm{e}^{\Omega_it}\cos(K_0x)}{\phi_0+2U_0\mathrm{e}^{\Omega_it}\cos(K_0x)}\right].
\end{eqnarray} 
Note that since $U_0, W_0\ll 1$, equations (\ref{eq3.77}) and (\ref{eq3.79}) 
can be approximated as:
\begin{eqnarray}
\label{eq3.78}
R\approx \phi_0\left[1+\frac{2U_0}{\phi_0}\mathrm{e}^{\Omega_it}\cos(K_0x)\right], 
\end{eqnarray}
%
%
and 
\begin{eqnarray}
\label{eq3.80}
\varphi\approx\frac{2W_0}{\phi_0}\mathrm{e}^{\Omega_it}\cos(K_0x).
\end{eqnarray} 
Then, 
using 
(\ref{eq3.78}) and (\ref{eq3.80}), the perturbed plane wave 
(\ref{eq3.75}) can be 
written as 
\begin{eqnarray}
\label{eq3.81}
\phi(x,t)\approx \phi_0\left[1+\frac{2U_0}{\phi_0}\mathrm{e}^{\Omega_it}\cos(K_0x)\right]
\mathrm{e}^{\mathrm{i}\mathcal{R}\sin\chi}
\mathrm{e}^{\mathrm{i}\theta},
\end{eqnarray}
where 
$\mathcal{R}=\frac{2W_0}{\phi_0}\mathrm{e}^{\Omega_it}$ and $\chi=\frac{\pi}{2}-K_0x$. However, using the formula (see \cite{AS,GRary}), 
\begin{eqnarray*}
\mathrm{e}^{\mathrm{i}\mathcal{R}\sin\chi}=\sum_{n=-\infty}^{+\infty}J_n(\mathcal{R})\mathrm{e}^{\mathrm{i}n\chi},
\end{eqnarray*}
where $J_n$ denotes the Bessel function of the first kind of integer order $n$, we find that (\ref{eq3.81}) eventually is approximated as
\begin{eqnarray}
\label{eq3.82}
\phi(x,t)\approx \phi_0\left[1+\frac{2U_0}{\phi_0}\mathrm{e}^{\Omega_it}\cos(K_0x)\right]\sum_{n=-\infty}^{+\infty}J_n(\mathcal{R})\mathrm{e}^{\mathrm{i}(\theta +n\chi)},
\end{eqnarray}
where 
\begin{eqnarray}
\label{eq3.83}
\mathrm{e}^{\mathrm{i}(\theta +n\chi)}=\exp\left\{\mathrm{i}\left[kx-\omega t+n\left(\frac{\pi}{2}-K_0x\right)\right]\right\}.
\end{eqnarray}
It is evident from the expression of the solution (\ref{eq3.82})-(\ref{eq3.83}) that the wave numbers of the unstable plane wave solution are $k\pm nK_0$, $n\in\mathbb{Z}$, when  the choice $K_0\in B_I$ is made.  
Notice that the above calculation is valid even in the case
where $\Omega_i=0$, however in the latter case the 
wavenumbers $k\pm nK_0$ are not subject to growth.
\subsection{Numerical study 2: Modulation instability of plane waves} 
\begin{figure}
\label{fig4}
\begin{center}
    \begin{tabular}{c}
   \includegraphics[scale=0.5]{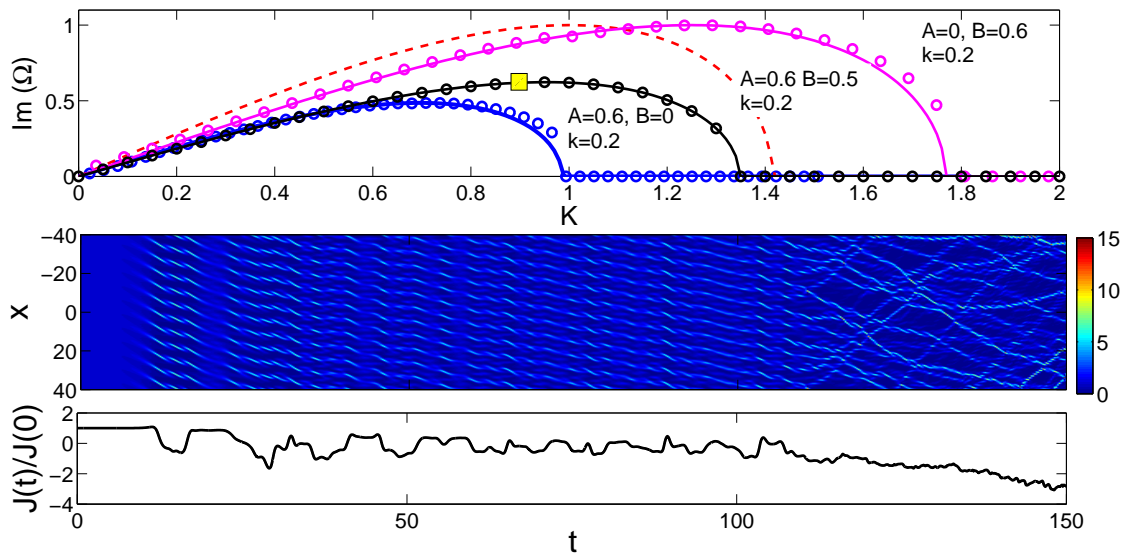} 
    \end{tabular}
\caption{\small{Top panel: 
The growth rate $\mathrm{Im}(\Omega)=\Omega_i$ as a function of the wavenumber $K$ 
of the perturbation, and the associated instability bands $B_I=[0,K_{\mathrm{max}}]$, 
for the ENLS equation (\ref{eq3.18}) with $s=+1$. 
The solid lines and circles correspond, respectively, to the analytical 
and numerical results, while the dotted (red) line corresponds to the 
focusing NLS limit with $A=B=0$. Other parameter values are:  
$k=0.2$, $\phi_0=1$.  
Middle panel:  Contour plot showing the evolution of the density $|\phi|^2$ 
as found numerically for 
$A=0.6, B=0.5$, and 
initial data as per (\ref{eq3.84}) with $k=0.2$, $\epsilon=2U_0=10^{-3}$ and $K=0.9$ (i.e., pertaining to the square of the top panel). 
Bottom panel: Evolution of the normalized functional 
$\mathbf{J}(t)/\mathbf{J}(0)$. 
}}
\end{center}
\end{figure}

In this section, we numerically integrate equation (\ref{eq3.18}), using as an initial condition the
perturbed plane-wave: 
%
\begin{eqnarray}
\label{eq3.84}
\phi(x,0)\approx\phi_0\left[1+\frac{2U_0}{\phi_0}\cos(K_0x)\right]\mathrm{e}^{\mathrm{i}kx},
\end{eqnarray}
%
with  $k\neq 0$ and parameters $\phi_0=1$ and $2U_0=\epsilon=10^{-3}$. 
We then calculate the growth rate of the amplitude of the initial plane wave, and compare it with the
above analytical results. 
Figure~3 
presents the comparison of numerical computations with the 
ENLS equation (\ref{eq3.18}) and $s=1$, for various values of the parameters $A,B\in\mathbb{R}$ and for fixed $k=0.2$. The upper panel of Fig.~3  
shows the analytically computed curves $\mathrm{Im}(\Omega)(K)=\Omega_i(K)$, and the relative instability bands $B_I=[0,K_{\mathrm{max}}]$ against their numerically computed counterparts. The analytical curves $\Omega_i(K)$ 
are depicted by solid lines, while the respective numerical results 
are shown 
by circles. The dotted (red) line corresponds to the analytical
result for the focusing NLS-limit ($A=B=0$). The 
solid lines correspond to the ENLS equation for 
$A=0, B=0.6$ [upper (magenta) line], $A=0.6, B=0.5$ [middle (black)
line] and $A=0.6, B=0$ [lower (blue) line]. In all cases  we observe the excellent agreement of the analytical 
prediction for the growth rates of the plane wave with the numerical results. 
Another interesting observation is the progressive increase of the length of the MI-bands,
from the limit 
of $A\neq 0$ and $ B=0$, 
i.e., from the NLS equation with the self-steepening effect,
to the limit of $A=0, B\neq0$,  
i.e., to the NLS equation with the third-order dispersion effect  
(see, e.g.~\cite{CPag95} for a discussion of these effects on the NLS model).
It is clear that if $1-3k B>0$ in (\ref{eq3.18}), increase of 
$A$ will have a stabilizing effect, while for $1-3k B<0$, it will
have a destabilizing effect. The role of
$B$ is, roughly, reversed with respect to that of $A$, although the associated
functional
dependence is structurally somewhat more complex.

The contour plot in the middle panel of Fig.~3 shows the evolution of the density $|\phi|^2$ 
for a case where both $A$ and $B$ are nonzero, namely for  
$A=0.6, B=0.5$; here,  
the wavenumbers of the plane wave and of the perturbation are chosen as $k=0.2$ and 
$K=0.9$, respectively. The 
growth rate for this case, $\Omega_i(0.9)$, is marked with the (yellow) 
square in the 
top panel of Fig.~3, in the relevant 
curve. The contour plot illustrates the manifestation of the MI and the concomitant dynamics.
In particular, at the initial stage of the evolution, one can observe the 
formation of an almost periodic pattern characterized by the excitation of the unstable wavenumber $K=K_0=0.9$, and at a later stage (for $t \gtrsim 120$) the   
formation of traveling localized structures induced by the manifestation of MI. 

It is also interesting to make the following observation. 
In the bottom panel of Fig.~3, we show the evolution of the normalized 
functional $\mathbf{J}(t)/\mathbf{J}(0)$, for the parameter values used in the middle panel of the same figure 
(recall that $A\ne B$ in this case), i.e., in the case where MI manifests itself. It is observed that, 
prior to the onset of MI, the above functional 
assumes a constant value which, however, changes upon the onset of MI.
%
%
%
In particular, as seen in this panel, the initial value 
$\mathbf{J}(0)$ is preserved
until the appearance of the first pattern of localized structures: at that instance, 
$\mathbf{J}$ first decreases and, for later times, fluctuates.
This effect can be explained by the equation 
\begin{eqnarray}
\label{MIc1}
\frac{d}{dt}\mathbf{J}(\phi(t)+\frac{3}{2}\mathrm{Re}\left[\int_{\Omega}(|\phi|^2)_x|\phi_x|^2dx\right](\delta\sigma-\alpha\rho)=0.
\end{eqnarray}
derived in the proof of Lemma 
\ref{conslaws} and (\ref{coen4FIN}). 
If 
plane waves of the form of Eq.~(\ref{eq3.61})
are modulationally stable, 
then $|\phi|^2 \sim \phi_0^2$ and $|\phi|_x^2\sim 0$, for all times. Therefore, (\ref{MIc1}) implies that 
when $A\neq B$, this quantity should be (roughly)
conserved and, thus, $\frac{d}{dt}\mathbf{J}(\phi(t)=0$ in the 
modulationally stable regime. 
On the other hand, in the modulationally unstable regime, there exists a time 
$T_{\mathrm{MI}}$, such that the perturbation $\phi_1(x,t)$ should become significant for all 
$t\geq T_{\mathrm{MI}}$. Therefore, in the case $A\neq B$ and 
in the MI regime, 
$\mathbf{J}$ cannot be (nearly) 
conserved for all times, and $\frac{d}{dt}\mathbf{J}(\phi(t)\neq 0$, for all $t\geq T_{\mathrm{MI}}$. 
Thus, in this set of cases the NLS-based (energy) functional $\mathbf{J}(t)$ may be used as a 
diagnostic of MI for the nonintegrable case of $A\ne B$.
However, this effect is not apparent in the integrable case of $A=B$, where $\mathbf{J}$ 
is always conserved, by virtue of Eq.~(\ref{MIc1}) and for 
this reason, the non-conservation of the Hamiltonian 
energy $\mathbf{J}(t)$ cannot be used as a generic diagnostic tool for the 
detection of MI. In the latter case of $A=B$ (as well as for
$A \neq B$), other tools can be used to detect the instability,
such as, e.g., the functional
$\tilde{{\bf J}}(t)=\int \phi(x,t) \phi^{\star}(x+a,t) dx$ (for an arbitrary fixed $a$), 
by analogy to the quantity $\sum_j \phi_j(t) \phi_{j+1}^{\star}(t)$ used for discrete
NLS models. Here, for the latter case of $A=B$, we will resort to Fourier
space diagnostics that will also clearly allow us to detect the relevant
instability.
%

Next, in Fig.~4 we present 
numerical results for the 
ENLS equation (\ref{eq3.18}) with $s=-1$ (i.e., in 
defocusing settings). The solid lines (circles) correspond to the analytical (numerical) results 
for three different cases: 
(i) $A=0, B=0.6$, $k=0.72$, characterized by the ``large'' instability band $B_I$ (magenta); 
(ii) $A=1, B=0.5$, $k=0.5$, with the ``medium'' 
instability band 
(black); 
(iii) $A=0.6, B=0$, $k=1$, with the ``small'' instability 
band (black). Once again, we observe the excellent agreement 
between analytical and numerical results for the growth rates.
Notice that, in this figure, 
instability band for the defocusing NLS-limit does not exist since 
plane waves are modulationally stable in this limit.
The broadening of the instability band from the limits corresponding to the NLS with the self-steepening term 
($A=0$, $B\ne 0$) to the NLS with the third-order dispersion term ($A\ne0$, $B=0$) observed in Fig.~4, can also 
be observed in this case.
%

In the 
bottom panel of Fig.~4, we show the space-time evolution of the density $|\phi|^2$, and the 
manifestation of the MI, for the more general case of $A,B \ne 0$, i.e., for parameter values $A=1, B=0.5$, as well as 
for $k=0.5$ and for an unstable wavenumber $K=K_0=1$.
The perturbation growth $\Omega_i(1)$ is marked with the (green) 
square in the 
top panel of the same figure 
in the respective 
$\Omega_i(K)$ curve. The initial stage of the dynamics is qualitatively similar to that in the middle panel of 
Fig.~3, but the later stage (for times $t \gtrsim 120$) is not characterized by the emergence of travelling localized 
structures: in the specific parameter regime, we have checked (results not shown here) 
that such structures are not supported by the system, contrary to the case of Fig.~3 where such states do exist.

For completeness, in Fig.~5 we show the Fourier transform $|\Phi(k)|$, of a modulationally unstable solution,
corresponding to the ENLS equation in the focusing nonlinearity case of $s=1$, with parameters $A=B=0.6$, $k=0.2$ and $K_0=0.9$; shown is the Fourier spectrum at a time instant after the instability has set in. The figure clearly
demonstrates
the fact that once MI manifests itself, 
wave numbers $k\pm n K_0$, $n\in \mathbb{N}$ are generated, 
as was illustrated above.

\begin{figure}
\label{fig5}
\begin{center}
    \begin{tabular}{c}
   \includegraphics[scale=0.5]{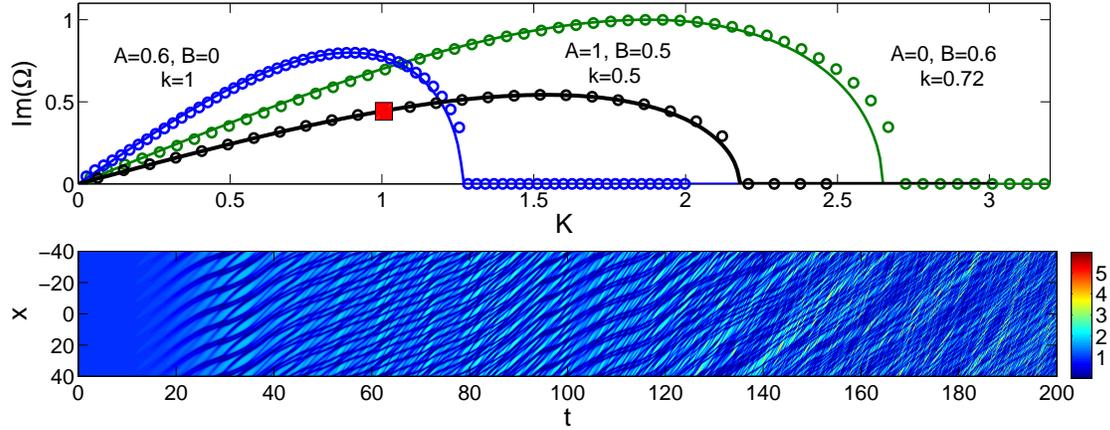} 
    \end{tabular}
\caption{\small{Top and bottom panels are same as corresponding ones in Fig.~3, 
but for the case of defocusing nonlinearity $s=-1$, and 
for $A=1, B=0.5$, 
$k=0.5$, $\epsilon=2U_0=10^{-3}$ and $K=1$. 
}}
\end{center}
\end{figure}
\begin{figure}
\label{fig6}
\begin{center}
    \begin{tabular}{c}
   \includegraphics[scale=0.65]{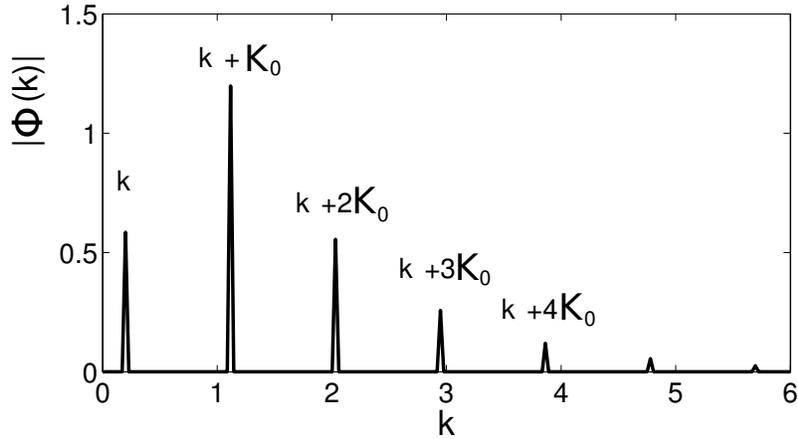} 
    \end{tabular}
\caption{\small{Fourier spectrum $\Phi(k,t=50)$ of the unstable 
plane wave solution in the focusing case $s=1$, for
parameters  $A=B=0.6$, and perturbed plane wave initial data (\ref{eq3.84}) with $\phi_0=1$, $k=0.2$, $\epsilon=2U_0=10^{-3}$ and $K_0=0.9$; one can observe the excitation of the wave numbers $k+nK_0$. }}
\end{center}
\end{figure}
%
%
\section{Discussion and conclusions}
\label{SECTION_VI}
In this work, analytical studies corroborated by numerical simulations, considered various aspects 
of the solutions of the ENLS equation (\ref{introeq1}) 
which, under 
the balance condition (\ref{cruc2008}), 
corresponds to the 
integrable Hirota equation. The ENLS equation is an important model finding applications in 
various mathematical and physical settings; these include 
geometric evolution equations, 
the evolution of vortex filaments (in the integrable focusing case), and propagation of short pulses in 
nonlinear optical fibers and nonlinear metamaterials.


Firstly, the global existence of $H^2$-weak solutions for the periodic initial-boundary value problem has been 
established. The global well-posedness in this regularity class, was an 
application of the justification of an $H^2$-type ``energy equation'' and the conservation laws 
satisfied by the weak solutions. 
Such questions are of importance, since the model may fall in the class of dispersive equations with 
lower-order non-linearities, but for which the periodic initial-boundary value problem may, in principle, exhibit finite 
time singularities. It was shown that the balance condition on the coefficients, and the assumption 
of non-vanishing second-order dispersion, are sufficient to guarantee the global-in-time solvability 
for the prescribed class of initial data.  

Next, we investigated the existence of exact travelling wave solutions, and focused --in particular-- 
on the existence of dark and bright solitary waves. 
The reduction of the problem to an ODE in the traveling wave frame
revealed that the balance 
condition (\ref{cruc2008}) is essential; in fact, it was found that it coincides with 
the compatibility assumptions needed to be satisfied by the system of 
ODE's governing the envelope of the desired travelling wave solution. 
The consideration 
of the reduced 2nd-order conservative dynamical system for the envelope, led to 
parameter regimes (including the frequency and the wavenumber of the travelling wave) for the existence 
of bright and dark solitary wave solutions; these soliton solutions are respectively given by:
\begin{eqnarray*}
	\phi(x,t)&=&\pm\sqrt{2\lambda}\,\mathrm{sech}\left[\sqrt{\lambda}\,(x-\upsilon t)\right]\mathrm{e}^{\mathrm{i}(k_0 x-\omega_0 t)},\;\;\lambda>0,\\
	\phi(x,t)&=&\pm \sqrt{|\lambda|}\mathrm{tanh}\left[\sqrt{\frac{|\lambda|}{2}}(x-\upsilon t)\right]\mathrm{e}^{\mathrm{i}(k_0 x-\omega_0 t)},\;\;\lambda<0,
\end{eqnarray*}
%
(cf. definitions and conditions on parameters $\lambda, \omega_0, k_0, \upsilon$ 
in Section \ref{SECTION_IV}). 
Importantly, we found that standing wave solutions, corresponding to $k_0=v=0$, do not exist --contrary to 
the case of the NLS model. Additionally, it was interesting to find that, apart from the balance condition, 
new conditions on dispersion and nonlinearity coefficients should be satisfied in order for the above travelling 
solitons to exist. For instance, we found that solitons may 
also exist at the zero dispersion point (where 2nd-order dispersion vanishes), in which case the ENLS is 
reduced to the cmKdV equation.

We have also presented results of direct numerical simulations 
of the ENLS equation, with initial data corresponding 
to randomly perturbed 
bright and dark solitons. Our results, corresponding to the integrable limit, 
indicate that the analytically determined soliton solutions 
are robust under perturbations.
On the other hand, numerical results corresponding to the case where 
the balance condition is not satisfied, have shown that bright and dark soliton initial conditions 
evolve to corresponding 
localized pulses that continuously emit 
radiation. 
Some of the results -- and particularly those corresponding to bright solitons -- were found to be 
in accordance with the findings of Refs.~\cite{GromTal2000,YangPeli2,Yang1}. Nevertheless, 
important differences were identified between the evolution of bright and dark pulses, concerning the radiation emission dynamics and the direction of their propagation.


Finally, we carried out a detailed analysis 
of the modulation instability (MI) of the plane wave solutions of the ENLS, namely, 
$$\phi(x,t)=\phi_0\mathrm{e}^{\mathrm{i}(kx-\omega t)},$$ and 
identified crucial differences 
between the ENLS equation and the corresponding NLS limit. For travelling plane waves, with wave number $k\neq 0$, 
it was shown that modulation instability 
can occur for both the focusing ($s=+1$) and the defocusing ($s=-1$) ENLS equation. The latter 
result is in contrast with the defocusing NLS limit, where such plane waves are always modulationally stable. 
In the case of stationary plane waves, with wave number $k=0$, the focusing ENLS equation 
possesses the same properties with its NLS counterpart, in the sense that MI conditions (and, as a result, 
the instability bands) coincide.
Also, interesting properties were revealed 
concerning the dependence of the 
instability bands 
on the higher-order effects.
These properties refer to the NLS equation with the self-steepening effect (corresponding to the case $A=0, B\neq 0$ of the ENLS 
(\ref{eq3.18})) and the NLS equation with the 3rd-order dispersion (corresponding to the case $A\neq 0, B=0$ of the ENLS 
(\ref{eq3.18})). It was found that the ENLS equation with $A,B\neq 0$ possesses a 
MI band of intermediate length between the respective MI 
band of the self-steepening limit (having the smallest MI band), 
and the 3rd-order dispersion limit (having the largest MI band). 
%

The results presented in this paper may pave the way for future work in many interesting directions. 
One such direction is to investigate the dynamics of weakly smooth initial data, when the sufficient conditions for global existence are not satisfied; is such a case, it would be relevant to seek for parameteric regimes in which possible instabilities, and even collapse, may emerge.
Another direction is to examine the applicability of our results --especially those pertaining to the defocusing case-- 
in the context of general curve evolution problems and geometric evolution equations. 
One could follow relevant studies on the geometric characterization of the defocusing NLS~\cite{DingInoguchi} 
or for the nonintegrable case~\cite{Onodera1,Onodera3}.
Additionally, regarding the stability of the solutions, it would be interesting to extend the program developed in
\cite{Cal2,Stephane2} (for studying the stability of closed solutions of the NLS equation and their correspondence to the
vortex filament motion) to the integrable focusing ENLS (Hirota) equation. Such studies are currently in progress and will 
be presented in future publications. 
%
%
%
%

\appendix
\section{Approximation scheme for the derivation of the energy equations}
\label{SECTION_II}
\par
In this complementary section, we include for the sake of completeness, the details of the proof of Lemma \ref{weaksol1}. For the justification of various computations needed for the derivation of the energy equation, we recall some
useful lemmas concerning time differentiation of Hilbert space
valued functions \cite{RTem2,XWang}.
%
%
\begin{lem}\label{xwa1}
	Let $V_i\subset H\subset V_i'$, $i=1,2$ a sequence of real Hilbert
	spaces, with continuous inclusions. If $\phi_i\in L^2([0,T],V_i)$,
	$\partial_t\phi_1\in L^2([0,T],V_2')$ and $\partial_t\phi_2\in
	L^2([0,T],V_1')$ then
	\begin{equation*}
		\partial_t(\phi_1,\phi_2)=<\phi_2,\partial_t\phi_1>_{{\ssy V_2,V_2'}}
		+<\phi_1,\partial_t\phi_2>_{{\ssy V_1,V_1'}}.
	\end{equation*}
\end{lem}
%
%
%
%
%
\begin{lem}\label{xwa2}
	Let $V\subset H\subset V'$ be three real Hilbert spaces, with
	continuous inclusion. We assume that $V$ is a generalized Banach
	algebra and
	\begin{equation*}
		\begin{gathered}
			<\phi,\psi>_{{\ssy V,V'}}=(\phi,\psi)_{\ssy
				H},\quad\forall\,\phi\in V,
			\quad\forall\,\psi\in H,\\
			<\phi,z\,\psi>_{{\ssy V,V'}}=<\phi\,\psi,z>_{{\ssy V,V'}},
			\quad\forall\,\phi,\psi\in V,\quad\forall\,z\in V'.
		\end{gathered}
	\end{equation*}
	Moreover, we assume that $\psi_i\in L^2([0,T],V)$ and
	$\partial_t\psi_i\in L^2([0,T],V')$ for $i=1,2$.
	Then the (weak) derivative $\partial_t(\psi_1\psi_2)$ exists and
	\begin{equation*}
		<\phi,\partial_t(\psi_1\psi_2)>_{{\ssy
				V,V'}}=<\phi,\partial_t\psi_1\,\psi_2
		+\psi_1\,\partial_t\psi_2>_{{\ssy V,V'}},
		\quad\forall\,\phi\in V.
	\end{equation*}
\end{lem}
Now, let us consider the unique solution $\phi\in
C([0,T],H^3_{per}(\Omega))\cap C^1([0,T], L^2(\Omega))$ of
\eqref{introeq1}-\eqref{introeq2}-\eqref{bc}. From Lemmas \ref{xwa1} and \ref{xwa2}, it follows that (see also
 \cite[Proposition 2.1, p. 170, eq. (2.17)]{XWang})
 \begin{equation}\label{devfunct}
 	\begin{gathered}
 		\phi_{txx}\in L^2([0,T], H^{-2}_{per}(\Omega)),\quad
 		\phi_{tx}\in L^2([0,T], H^{-1}_{per}(\Omega)),\\
 		\phi_{xxxx}\in L^2([0,T],H^{-1}_{per}(\Omega)),\\
 		\phi_{txx}=\phi_{xtx}=\phi_{xxt},\quad
 		\phi_{tx}=\phi_{xt}.
 	\end{gathered}
 \end{equation}
Accordingly, the infinitely smooth approximating function $v_n$ converges to the solution $\phi$ in the spaces 
\begin{equation}
	\begin{array}{lllll}
		\label{propvn}
		v_n\rightharpoonup \phi,\;\;\mbox{in}\;\;L^2([0,T], H^{3}_{per}(\Omega)),
		&&\;v_{nx}\rightharpoonup \phi_x,\;\;\mbox{in}\;\;L^2([0,T], H^{2}_{per}(\Omega)),\\
		v_{nt}\rightharpoonup \phi_t,\;\;\mbox{in}\;\;L^2([0,T], L^2(\Omega)),
		&&\;v_{nxx}\rightharpoonup \phi_{xx},\mbox{in}\;\;L^2([0,T], H^{1}_{per}(\Omega)),\\
		v_{ntx}\rightharpoonup \phi_{tx},\;\;\mbox{in}\;\;L^2([0,T], H^{-1}_{per}(\Omega)),
		&&\;v_{nxxx}\rightharpoonup \phi_{xxx},\mbox{in}\;\;L^2([0,T], L^2(\Omega)),\\
		v_{ntxx}\rightharpoonup \phi_{txx},\;\;\mbox{in}\;\;L^2([0,T], H^{-2}_{per}(\Omega)),
		&&\;v_{nxxxx}\rightharpoonup\phi_{xxxx},\mbox{in}\;\;L^2([0,T], H^{-1}_{per}(\Omega)),
	\end{array}
\end{equation}
as a consequence of Theorem~\ref{thmloc}. Now,  we proceed to approximate the time derivative of the second term of of the functional (\ref{pseudocons}) as follows:
\begin{equation}\label{multip1}
\begin{split}
\tfrac{d}{dt}\left( |v_{nx}|^2,|v_n|^2\right)_{\ssy L^2}=&\, \ld |
v_n|^2,\tfrac{d}{dt}|v_{nx}|^2\rd_{\ssy H^2,H^{-2}}
+\ld | v_{nx}|^2,\tfrac{d}{dt}|v_{n}|^2\rd_{\ssy H^2,H^{-2}}\\
=&\,2\,\ld|v_n|^2,\overline{v}_{nx}v_{nxt}\rd_{\ssy H^2,H^{-2}}
+2\,\ld |v_{nx}|^2, \overline{v}_n v_{nt}\rd_{\ssy H^2,H^{-2}}\\
=&\,2\,\ld |v_n|^2\overline{v}_{nx},v_{nxt}\rd_{\ssy H^2,H^{-2}}
+\ld |v_{nx}|^2\overline{v}_n,v_{nt}\rd_{\ssy H^2,H^{-2}}\\
=&\,-2\,\ld |v_n|^2\overline{v}_{nxx},v_{nt}\rd_{\ssy H^1,H^{-1}}
-2\ld\overline{v}^2_{nx} v_n,v_{nt}\rd_{\ssy H^2,H^{-2}}\\
&\quad-2\ld |v_{nx}|^2\overline{v}_n,v_{nt}\rd_{H^2,H^{-2}}
+2\ld |v_{nx}|^2\overline{v}_n,v_{nt}\rd_{\ssy H^2,H^{-2}}\\
=&\,-2\,\ld |v_n|^2\overline{v}_{nxx},v_{nt}\rd_{\ssy
	H^1,H^{-1}}-2\ld
\overline{v}_{nx}^2 v_n,v_{nt}\rd_{\ssy H^2,H^{-2}}.\\
\end{split}
\end{equation}
Integration of (\ref{multip1}) with respect to time, gives (\ref{prepassage}), from which we are passing to the limit equation (\ref{multip1b}), as discussed in the Lemma \ref{weaksol1}. Proceeding further, the substitution of $\phi_t$ into (\ref{multip1b}), in terms of its expression given by the pde
\eqref{in1}, can be handled with the computations
%
%
\begin{equation}\label{multip1c}
\begin{split}
-2\ld|\phi|^2\,\overline{\phi}_{xx},\phi_{t}\rd_{\ssy
	H^1,H^{-1}}=&\,-2\,\ld |\phi|^2\,\overline{\phi}_{xx},
-3\,\alpha\,|\phi|^2\,\phi_x+{\rm i}\,\rho\,\phi_{xx}
-\sigma\,\phi_{xxx}
+{\rm i}\,\delta\,|\phi|^2\,\phi\rd_{\ssy H^1,H^{-1}}\\
=&\,6\,\alpha\,\ld\overline{\phi}_{xx}\,\phi_x,|\phi|^4\rd_{\ssy
	H^1,H^{-1}}
+2\,\sigma\,\ld|\phi|^2\,\overline{\phi}_{xx},\phi_{xxx}\rd_{\ssy H^1,H^{-1}}\\
&\quad-2\,\delta\,\ld\overline{\phi}_{xx},
{\rm i}\,|\phi|^4\,\phi\rd_{\ssy H^1,H^{-1}},\\
\ld\overline{\phi}_{x}^2\,\phi,\phi_{t}\rd_{\ssy H^2,H^{-2}}
=&\,6\,\alpha\,\ld\overline{\phi}_{x}^2\,
\phi,|\phi|^2\,\phi_x\rd_{\ssy H^2,H^{-2}}-2\rho\ld
\overline{\phi}_{x}^2\,\phi,
\phi_{xx}\rd_{\ssy H^2,H^{-2}}\\
&\,+2\,\sigma\,\ld\overline{\phi}_{x}^2\, \phi,\phi_{xxx}\rd_{\ssy
	H^2,H^{-2}}-2\,\delta\,\ld\overline{\phi}_{x}^2\,
\phi,{\rm i}\,|\phi|^2\,\phi\rd_{\ssy H^2,H^{-2}}.\\
\end{split}
\end{equation}
Now, by using \eqref{Sembe}, \eqref{multip1b}, and \eqref{multip1c}, we
derive the equation
\begin{equation}\label{multip2}
\begin{split}
\tfrac{d}{dt}\int_{\ssy\Omega}|\phi_x|^2|\phi|^2\,dx=&\,3\,\alpha
\,\mathrm{Re}\left[\int_{\ssy\Omega}|\phi|^4\,(|\phi_x|^2)_x\;dx\right]
+\sigma\,\mathrm{Re}\left[\int_{\ssy\Omega}
|\phi|^2\,(|\phi_{xx}|^2)_x\;dx\right]\\
&+2\,\delta\,\mathrm{Im}\left[\int_{\ssy\Omega}|\phi|^4\,\phi
\,\overline{\phi}_{xx}\;dx\right]
+3\,\alpha\,\mathrm{Re}\left[\int_{\ssy\Omega}|\phi|^2
\,|\phi_x|^2\,(|\phi|^2)_x\;dx\right]\\
&+2\,\sigma\,\mathrm{Re}\left[\int_{\ssy\Omega}
\overline{\phi}_x^2\,\phi\,\phi_{xxx}\;dx\right]
+2\,\rho\,\mathrm{Im}\left[\int_{\ssy\Omega}\overline{\phi}_x^2
\,\phi\,\phi_{xx}\;dx\right]\\
&+2\,\delta\,\mathrm{Im}\left[\int_{\ssy\Omega}\overline{\phi}_x^2
\,\phi^2|\phi|^2\;dx\right],\\
\end{split}
\end{equation}
being valid for the solution $C([0,T],H^3_{per}(\Omega))\cap C^1([0,T], L^2(\Omega))$.  
%
%
%
%
\par
We compute next, the time derivatives for the remaining terms of the functional (\ref{pseudocons}).   For the term 
$(\phi^2,\overline{\phi}_x^2)_{\ssy L^2(\Omega)}$, by using the approximating
sequence $v_n$,  and Lemmas \ref{xwa1} and \ref{xwa2},  we observe
that
\begin{equation}\label{multip4a}
	\begin{split}
		\tfrac{d}{dt}\left(v_n^2,\overline{v}^2_{nx}\right)_{\ssy
			L^2(\Omega)}&=\ld\overline{v}^2_{nx},\tfrac{d}{dt}v_n^2\rd_{\ssy
			H^2,H^{-2}}+
		\ld v_n^2,\tfrac{d}{dt}\overline{v}_{nx}^2\rd_{\ssy H^2,H^{-2}}\\
		&=2\,\ld\overline{v}^2_{nx}v_n,v_{nt}\rd_{\ssy H^2,H^{-2}}
		+2\,\ld \overline{v}^2v_{nx},v_{nxt}\rd_{\ssy H^2,H^{-2}}\\
		&=2\,\ld\overline{v}^2_{nx}v_n,v_{nt}\rd_{\ssy H^2,H^{-2}}
		-2\,\ld\partial_x(\overline{v}_n^2v_{nx}),v_{nt}\rd_{\ssy H^1,H^{-1}}\\
		&=2\,\ld\overline{v}^2_{nx}v_n,v_{nt}\rd_{\ssy H^2,H^{-2}}-4\ld
		\overline{v}_n|v_{nx}|^2,v_{nt}\rd_{\ssy H^2,H^{-2}}- 2\ld
		\overline{v}_n^2v_{nxx},v_{nt}\rd_{\ssy H^1,H^{-1}}.\\
	\end{split}
\end{equation}
We continue by using the same arguments, as for the derivation of
\eqref{multip1b}. Letting $n\rightarrow\infty$ we get that
\begin{equation}\label{multip4}
	\tfrac{d}{dt}(\phi^2,\overline{\phi}_x^2)_{\ssy
		L^2(\Omega)}=\mathbf{I_1}+\mathbf{I_2}+ \mathbf{I_3},
\end{equation}
where the terms $(\mathbf{I_k})_{k=1}^3$ are found to be
\begin{equation}\label{intt1}
	\begin{split}
		\mathbf{I_1}=&\,2\,\ld\phi\,\overline{\phi}_x^2,\phi_t\rd_{\ssy
			H^{2},H^{-2}}-3\,\alpha\,\ld|\phi|^2\,|\phi_x|^2,\partial_x
		\,(|\phi|^2)\rd_{\ssy H^{2},H^{-2}}\\
		&+2\,\rho\,\ld\phi\,\overline{\phi}_x^2,{\rm
			i}\,\phi_{xx}\rd_{\ssy
			H^{2},H^{-2}}-2\,\sigma\,\ld\phi\,\overline{\phi}_x^2,
		\phi_{xxx}\rd_{\ssy H^{2},H^{-2}}
		+2\,\delta\,\ld\overline{\phi}_x^2,
		{\rm i}|\phi|^2\,\phi^2\rd_{\ssy H^{2},H^{-2}},\\
	\end{split}
\end{equation}
\begin{equation}\label{intt2}
	\begin{split}
		\mathbf{I_2}=&-4\,\ld\overline{\phi}\,|\phi_x|^2,\phi_t\rd_{\ssy
			H^{2},H^{-2}} +6\,\alpha\,\ld|\phi|^2\,|\phi_x|^2,
		\partial_x(|\phi|^2)\rd_{\ssy H^{2},H^{-2}}\\
		&-4\,\rho\,\ld|\phi_x|^2\,\overline{\phi},{\rm
			i}\,\phi_{xx}\rd_{\ssy H^{2},H^{-2}}
		+4\,\sigma\,\ld|\phi_x|^2\,\overline{\phi},
		\phi_{xxx}\rd_{\ssy H^{2},H^{-2}},\\
	\end{split}
\end{equation}
and
\begin{equation}\label{intt3}
	\begin{split}
		\mathbf{I_3}=&-2\,\ld \overline{\phi}^2\phi_{xx},\phi_t\rd_{\ssy
			H^{1},H^{-1}}+6\,\alpha\,\ld|\phi|^2\,\phi_x\,
		\overline{\phi}^2,\phi_{xx}\rd_{\ssy H^{2},H^{-2}}\\
		&-2\,\rho\,\ld \overline{\phi}^2,{\rm i}\,\phi_{xx}^2\rd_{\ssy
			H^{3},H^{-3}} +2\,\sigma\,\ld
		\overline{\phi}^2\phi_{xx},\phi_{xxx}\rd_{\ssy
			H^{1},H^{-1}}-2\,\delta\,\ld|\phi|^4\,\overline{\phi}, {\rm
			i}\,\phi_{xx}\rd_{\ssy H^{3},H^{-3}}.\\
	\end{split}
\end{equation}
The resulting  equation, involves the time derivative of
$\|\phi_{xx}\|_{\ssy L^2(\Omega)}$. We observe first that due to
periodicity, it holds that
\begin{equation}\label{multip3Aa}
	\begin{split}
		2\ld v_{nt}-\mathrm{i}\,\rho\, v_{nxx}+\sigma\,
		v_{nxxx},\partial^4_xv_n\rd_{\ssy H^2,H^{-2}}=&\,2\,\ld
		v_{nt},\partial^4_xv_n\rd_{\ssy H^2,H^{-2}}-2\ld\mathrm{i}\,\rho\,
		v_{nxx},\partial^4_xv_n\rd_{H^2,H^{-2}}\\
		&+2\,\ld \sigma\, v_{nxxx},\partial^4_xv_n\rd_{\ssy H^2,H^{-2}}\\
		=&\,\ld v_{nxxt},v_{nxx}\rd_{\ssy H^2,H^{-2}}\\
		=&\,\tfrac{1}{2}\,\tfrac{d}{dt}\|v_{nxx}\|_{\ssy L^2(\Omega)}^2.\\
	\end{split}
\end{equation}
Furthermore, since the right-hand side of \eqref{in1} lies in
$L^{2}_{loc}(\mathbb{R},H^2_{per}(\Omega))$ (due to \eqref{prop1}
and \eqref{prop2}), the left-hand side lies in
$L^{2}_{loc}(\mathbb{R}, H^2_{per}(\Omega))$. Thus, it is justified
to pass to the limit in \eqref{multip3Aa}, which converges as
$n\rightarrow\infty$, to
\begin{equation}\label{multip3Ab}
	2\ld \phi_{t}-\mathrm{i}\,\rho\,\phi_{xx}
	+\sigma\,\phi_{xxx},\partial^4_x\phi\rd_{\ssy H^2,H^{-2}}
	=\tfrac{1}{2}\,\tfrac{d}{dt}\|\phi_{xx}\|_{\ssy L^2(\Omega)}^2.
\end{equation}
Then, by substitution of the right-hand side of \eqref{in1} to the
right-hand side of \eqref{multip3Ab}, we get
\begin{equation}\label{multip3b}
	\tfrac{d}{dt}\|\phi_{xx}\|^2_{\ssy
		L^2(\Omega)}=-6\,\alpha\,\ld|\phi|^2\phi_{x},
	\partial_x^4\overline{\phi}\rd_{\ssy H^{1},H^{-1}}
	+2\,\delta\,\ld{\rm i}\,|\phi|^2\,\phi,
	\partial_x^4\overline{\phi}\rd_{\ssy H^{1},H^{-1}}.
\end{equation}
Again, due to periodicity, we have
\begin{equation}\label{multip3c}
	\begin{split}
		-6\,\alpha\,\ld |\phi|^2\,\phi_{x},
		\partial_x^4\overline{\phi}\rd_{\ssy H^{1},H^{-1}}
		=&\,6\,\alpha\,\ld\partial_x(|\phi|^2\,\phi_{x}),
		\overline{\phi}_{xxx}\rd_{\ssy H^{1},H^{-1}}\\
		=&\,6\,\alpha\,\ld|\phi_x|^2\,\phi,\overline{\phi}_{xxx}\rd_{\ssy
			H^{2},H^{-2}}+6\,\alpha\,\ld\phi_x^2\,\overline{\phi},
		\overline{\phi}_{xxx}\rd_{\ssy H^2,H^{-2}}\\
		&\,+6\,\alpha\ld|\phi|^2\,\phi_{xx},
		\overline{\phi}_{xxx}\rd_{\ssy H^{2},H^{-2}}.\\
	\end{split}
\end{equation}
The terms containing weak derivatives of third and fourth order
are reduced as follows:
\begin{gather}
	2\,\sigma\,\ld{\phi}^2\,\phi_{xx},\phi_{xxx}\rd_{\ssy
		H^{1},H^{-1}}=-2\,\sigma\,\ld\phi_{xx}^2,\overline{\phi}
	\,\overline{\phi}_x\rd_{\ssy H^{1},H^{-1}},\label{intt4}\\
	6\,\alpha\,\ld|\phi_x|^2,\overline{\phi}_{xxx}\rd_{\ssy
		H^{1},H^{-1}}=-6\,\alpha\,\ld\overline{\phi}\,\overline{\phi}_x,
	\phi_{xx}^2\rd_{\ssy H^{2},H^{-2}}
	-3\,\alpha\,\ld|\phi_{xx}|^2,(|\phi|^2)_x\rd_{\ssy H^{1},H^{-1}},
	\label{intt5}\\
	6\,\alpha\,\ld\phi_x^2\,\overline{\phi},\overline{\phi}_{xxx}\rd_{\ssy
		H^{2},H^{-2}}=-6\,\alpha\,\ld|\phi_{xx}|^2,(|\phi|^2)_x\rd_{\ssy
		H^{1},H^{-1}},\label{intt6}
\end{gather}
and
\begin{equation}\label{intt7}
	\begin{split}
		2\,\delta\,\ld {\rm
			i}\,|\phi|^2\phi,\overline{\phi}_{xxxx}\rd_{H^{1},H^{-1}}
		=&\,8\,\delta\ld{\rm i}\,|\phi_x|^2\,\phi,\overline{\phi}_{xx}\rd_{\ssy H^{2},H^{-2}}\\
		&\quad+4\delta \ld{\rm
			i}\,\phi_x^2\overline{\phi},\overline{\phi}_{xx}\rd_{H^{2},H^{-2}}
		+2\delta\ld{\rm i}\,\phi^2,\overline{\phi}_{xx}^2\rd_{\ssy
			H^{2},H^{-2}}.\\
	\end{split}
\end{equation}
After all these reductions, we may handle the nonlinear terms of
\eqref{intt4}-\eqref{intt7} with \eqref{prop1} and \eqref{prop2},
and justify the application of \eqref{Sembe}, to conclude with the proof of the Lemma \ref{weaksol1}.
\section*{Acknowledgements}
We would like to thank the referees for their constructive comments.

\end{document}